\documentclass[aps,prx,showpacs,twocolumn,reprint]{revtex4-2}

\usepackage{qcircuit}
\usepackage[dvips]{graphicx}
\usepackage{amsmath,amssymb,amsthm,mathrsfs,amsfonts,dsfont}
\usepackage{subfigure, epsfig}
\usepackage{braket}
\usepackage{hyperref}
\usepackage{bm}
\usepackage{enumerate}
\usepackage{color}
\usepackage{graphicx}
\usepackage{pgf}

\usepackage{amsthm}

\newtheorem{remark}{Remark}
\newtheorem{definition}{Definition}
\newtheorem{statement}{Statement}
\newtheorem{theorem}{Theorem}
\newtheorem{lemma}{Lemma}

\newcommand{\tr}{\mathrm{Tr}}

\newcommand{\var}{\mathrm{Var}}

\newcommand{\prob}{\mathrm{prob}_0}

\newcommand{\dmin}{\Delta_\mathrm{m}}

\newcommand{\qmin}{Q_\mathrm{m}}

\begin{document}

% --------------------  TITLE  --------------------

\title{The Dominant Eigenvector of a Noisy Quantum State}

% ------------  AUTHORS AND AFFILIATIONS ----------
\author{B\'alint Koczor}
\email{balint.koczor@materials.ox.ac.uk}
\affiliation{Department of Materials, University of Oxford, Parks Road, Oxford OX1 3PH, United Kingdom}

% --------------------  ABSTRACT  --------------------

\begin{abstract}
Although near-term quantum devices have no comprehensive solution for correcting errors, numerous techniques have been proposed for achieving practical value. Two works have recently introduced the very promising Error Suppression by Derangements (ESD) and Virtual Distillation (VD) techniques. The approach exponentially suppresses errors and ultimately allows one to measure expectation values in the pure state as the dominant eigenvector of the noisy quantum state. Interestingly this dominant eigenvector is, however, different than the ideal computational state and it is the aim of the present work to comprehensively explore the following fundamental question: how significantly different are these two pure states? The motivation for this work is two-fold. First, comprehensively understanding the effect of this coherent mismatch is of fundamental importance for the successful exploitation of noisy quantum devices. As such, the present work rigorously establishes that in practically relevant scenarios the coherent mismatch is exponentially less severe than the incoherent decay of the fidelity -- where the latter can be suppressed exponentially via the ESD/VD technique. Second, the above question is closely related to central problems in mathematics, such as bounding eigenvalues of a sum of two matrices (Weyl inequalities) -- solving of which was a major breakthrough. The present work can be viewed as a first step towards extending the Weyl inequalities to eigenvectors of a sum of two matrices -- and completely resolves this problem for the special case of the considered density matrices.
\end{abstract}

\maketitle

\section{Introduction}
Quantum devices can already prepare complex quantum states whose behaviour cannot be simulated using classical
computers with practical levels of resource~\cite{GoogleSupremacy,ChinaSupremacy}.
Sufficiently advanced quantum computers may have the potential to perform useful tasks of value to
society that cannot be performed by other means, such as simulating molecular systems \cite{SamReview}.
However, the early devices are incapable of error correction as required for fault-tolerant universal
systems that we expect to emerge eventually. Since the implementation of general quantum error
correcting codes (QECs) is prohibitively expensive,  the early machines do not have a comprehensive
solution to accumulating noise~\cite{preskill2018quantum}. Nevertheless, very promising applications
have been proposed for exploiting Noisy Intermediate-Scale Quantum (NISQ) devices:
variational quantum eigensolvers VQE and similar variants are expected to be able to solve important, 
practically relevant problems, such as finding ground states or optimising probe states for
quantum metrology and beyond \cite{farhi2014quantum,peruzzo2014variational,PRXH2,mcclean2016theory,
	PhysRevLett.118.100503,Li2017,PhysRevX.8.011021,Santagatieaap9646,kandala2017hardware,kandala2018extending,
	PhysRevX.8.031022,kokail2019self, koczor2019variational, koczor2019quantum, koczor2020quantumAnalytic}.
Refer also to the recent reviews \cite{endo2020hybrid,cerezo2020variationalreview,bharti2021noisy}.

The control of errors is thus fundamental to the successful exploitation of quantum devices
and numerous proposals have been put forward to mitigating errors in noisy machines.
These typically aim to learn the effect of imperfections on expectation values of observables
and try to predict their ideal, noise-free values. A very promising approach has recently been introduced
by two independent works and was named Error Suppression by Derangements (ESD \cite{koczor2020exponential})
and Virtual Distillation (VD \cite{huggins2020virtual}).
This technique prepares $n$ copies of the noisy quantum state and in turn allows to suppress errors
in expectation values exponentially when increasing $n$. The approach relies on the assumption that
the dominant eigenvector of a noisy quantum state, as modelled by a density matrix $\rho$, approximates
the ideal computational state.

This brings us to the core question of the present work.
Given a noisy quantum state $\rho$, how well does its dominant eigenvector approximate
the state that one would obtain from a perfect, noise-free computation?
It was already noted in refs.~\cite{koczor2020exponential,huggins2020virtual} that
even incoherent errors will in general introduce a drift in the dominant eigenvector.
This drift was named `coherent mismatch' and `noise floor' by the two works.
The aim of the present work is to comprehensively answer the above question by
deriving rigorous lower and upper bounds and scaling results.
The motivation for this work is two-fold.

First, the very promising ESD/VD approach crucially relies on the
above assumption that the dominant eigenvector is a good approximation of the ideal computational state.
However, a drift in the dominant eigenvector, the coherent mismatch, can crucially
influence the efficacy of the error suppression as illustrated in Fig.~\ref{plot_trdist}.
It is therefore vital for the successful exploitation of the technique to comprehensively
understand the drift in the dominant eigenvector. Fig.~\ref{plot_trdist} also shows that the previous
`pessimistic' upper bound $\sqrt{c}$ is quadratically reduced as $c$ if the aim is to prepare eigenstates
(see Sec.~\ref{precision_motivation}).
This is very encouraging since in fact most near-term quantum algorithms aim to prepare eigenstates
\cite{endo2020hybrid,cerezo2020variationalreview,bharti2021noisy}.

Second, understanding how noise affects quantum states is of fundamental
importance. While the mathematical formalism for describing noise processes has been
much investigated in the literature, there are still open questions.
Indeed, understanding noise in quantum systems is vital for
the successful exploitation of noisy quantum devices, however, the appropriate
modelling of quantum systems has significant implications in mathematics.
As such, the present work makes exciting connections to important problems in 
mathematics, such as bounding eigenvalues of a sum of two matrices and bounding
norms of commutators.

Let us now briefly summarise the most important results in relation to the above two points
ordering them thematically -- while a more detailed discussion of the results is presented in Sec.~\ref{discussion} 
that follows their order of appearance in the manuscript. 
In Sec.~\ref{results} we explicitly construct a family of worst/best-case extremal quantum states
that saturate the present upper and lower bounds of the coherent mismatch. These extremal states then allow us
to generally understand the coherent effect of incoherent noise channels in quantum systems and to argue about
the efficacy of the ESD/VD approach in complete generality
-- and prior perturbative approximations fail in this regime~\cite{koczor2020exponential,huggins2020virtual}
which is discussed in Sec.~\ref{sec_arrowhead_eigenvector}.
As such, in Sec.~\ref{numb_copies} we rigorously prove that even in the worst-case scenario one needs at least 3-4 copies in
practice to suppress incoherent errors to the level of the coherent mismatch: thus near term quantum
devices will be guaranteed to be oblivious to such coherent effects if they are limited in
preparing a large number of copies.

In Sec.~\ref{circuit_section} we analyse typical quantum circuits used in near-term quantum devices:
We derive guarantees that the coherent mismatch
decreases when increasing the size of the computation (even exponentially when increasing Rényi entropies
of the errors, see Sec.~\ref{sec_upbdelta}). We finally conclude that the coherent mismatch  is exponentially less severe when increasing
the circuit error rate than the severity of the incoherent decay of the fidelity -- where the latter can
be suppressed exponentially with the ESD/VD approach. 
We also prove in Sec.~\ref{sec_lower_upper_bounds} that our lower and upper bounds nearly coincide in the practically most important regions 
thus tightly confining the possible values the coherent mismatch can take up.

As mentioned above, the present work is closely related to important themes in mathematics, such as bounding
the eigenvalues of a sum of
two matrices (Weyl's inequalities) and we discuss these connections in Sec.~\ref{sec_related_maths}.
As such, the present work can be viewed as a first step towards extending Weyl's inequalities
for eigenvalues to the highly non-trivial case of the eigenvectors of a sum of two matrices -- and we present
a complete resolution of this problem for the special case of the considered density matrices.
Furthermore, another open question in mathematics was concerned with bounding the norm of a commutator
and this problem was only very recently solved
\cite{bottcher2008frobenius,vong2008proof,wu2010short,bottcher2005big,laszlo2007proof,cheng2010commutators,wenzel2010impressions}.
The present work significantly tightens those bounds for the special case of the considered density matrices
in Sec.~\ref{sec_comm_norm}.

We note that the following sections of the manuscript will gradually build on each other and
	the appearance of results might differ from the thematic ordering of the above summary.
Let us now introduce the core problem in more detail in Sec.~\ref{sec_prob_def}
and then recapitulate the most important notions in the context of the ESD/VD
approach in Sec.~\ref{sec_error_suppression}.

%-------------------------
\begin{figure}[tb]
	\begin{centering}
		\includegraphics[width=0.47\textwidth]{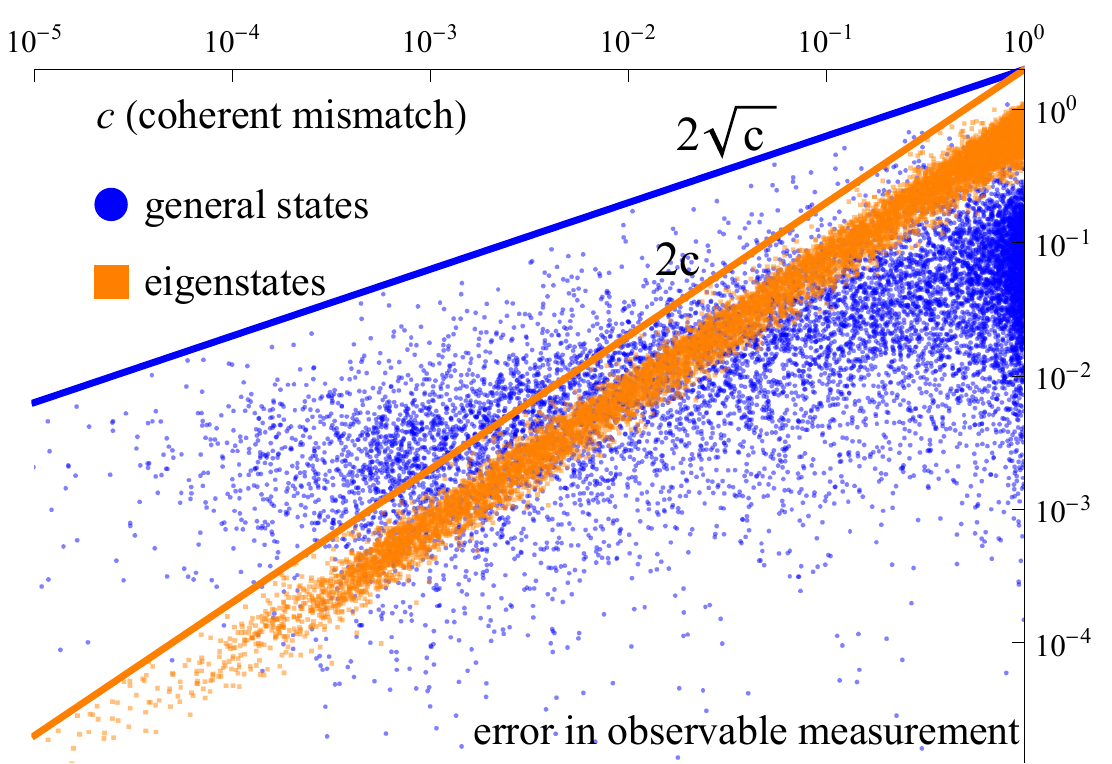}
		\caption{
			An important application of the present work is that it allows one to determine the ultimate
			precision of the ESD/VD error suppression technique \cite{koczor2020exponential,huggins2020virtual}.
			The trace distance (blue line) used in
			ref~\cite{huggins2020virtual} is given by the square root of the coherent mismatch $c$ from \cite{koczor2020exponential}
			and generally upper bounds the error $| \langle \psi_{id} | O |\psi_{id} \rangle - \langle \psi | O |\psi \rangle| \leq 2 \sqrt{c}$
			in estimating expectation values (with $\lVert O \rVert_{\infty}=1$) with
			the ideal computational state $|\psi_{id} \rangle$  vs. the dominant eigenvector $|\psi \rangle$.
			Most randomly generated quantum states (blue dots) are significantly below this pessimistic general bound  (blue line)
			-- which was already noted in ref~\cite{huggins2020virtual}. We show that this
			error bound  is quadratically smaller as $2 c$ (orange line)
			in the specific but pivotal case of preparing eigenstates
			(orange rectangles) -- the aim of most near-term quantum algorithms.
			Refer to Sec.~\ref{precision_motivation}.
			\label{plot_trdist}
		}
	\end{centering}
\end{figure}
%-------------------------

\subsection{Problem definition\label{sec_prob_def}}

Let us first introduce the most important notions used in this work.
Recall that a pure quantum state $| \psi_{id} \rangle$ is an element of a $d$-dimensional Hilbert
space. In an ideal quantum computation this quantum
state is prepared by a unitary quantum circuit (unitary transformation) as
$|\psi_{id} \rangle := U_{c} |\underline{0} \rangle$ that acts on a reference state.
The quantum circuit is typically decomposed into a product of (universal) gates
as $U_{c} = U_\nu \cdots U_2 U_1$.

In a realistic setting where the quantum gates are imperfect (or when the errors are not corrected)
the actual quantum state needs to be modelled by a density matrix $\rho := \Phi_{c} \rho_{\underline{0}}$
that is prepared via a CPTP~\cite{nielsenChuang} map $\Phi_{c}$.
For example, this noisy circuit is typically decomposed into a series
of individual noisy gates as $\Phi_{c} \approx  \Phi_\nu \cdots \Phi_2 \Phi_1$, but this in general
is only an approximation due to the presence of possible correlated noise.

Let us introduce another `representation of noise': We show in Appendix~\ref{def_state_appendix} that a
large class of density matrices admit the decomposition
$\rho =  \eta \rho_{id} + (1-\eta)  \rho_{err}$  for some constant $\eta > 0 $.
Here $\rho_{err}$ is a valid density matrix that can be interpreted as an error state
that occurs with probability $1-\eta$ and is incoherently superimposed (mixed) with the
ideal computational  state $\rho_{id}:=|\psi_{id} \rangle \langle \psi_{id} |$ which occurs
with probability $\eta$. 

Let us now consider a simple, but practically very important example to illustrate the previous point:
an error model $\Phi_{c} =  \Phi_\nu \cdots \Phi_2 \Phi_1$ in which errors happen
during the execution of an individual quantum gate with probability $\epsilon$ and thus
the corresponding Kraus-map representation of the $k^{th}$ noisy quantum gate can be defined as
\begin{equation} \label{error_channel}
	\Phi_{k} \,  \rho := (1-\epsilon) U_k \rho U_k^\dagger + \epsilon \sum_{j=1}^{K} M_{j k} \rho M_{jk}^\dagger.
\end{equation}
Here $M_{jk}$ corresponds to some (arbitrary) error event and $K$ determines the Kraus rank of the error model
while $U_k$ is the ideal unitary gate.
A large class of noise channels that are typically
used to model  errors in quantum circuits admit this form, for example dephasing, bit flip and depolarising errors
\cite{nielsenChuang}.
Within this error model we can straightforwardly obtain the decomposition in Eq.~\eqref{rho_decomp_sum}
into an ideal state $\rho_{id}$ and an error  density matrix via the probability $\eta = (1-\epsilon)^\nu$;
indeed the error matrix via
$(1-\eta)\rho_{err} = \Phi_{c} \rho_{\underline{0}} - \eta U_c \rho_{\underline{0}} U_c^\dagger $
can be shown to be a valid density matrix.
The completely general case is discussed in Appendix~\ref{def_state_appendix}.

Before stating our main problem, let us recall that a density matrix is a trace-class
operator with trace norm $\lVert \rho \rVert_1 = 1$ and trace $\tr \rho =1$,
and thus it can be written in terms of its spectral resolution as
\begin{equation}\label{rho_decomp}
	\rho := \lambda |\psi \rangle \langle \psi | + \sum_{k=2}^d \lambda_k |\psi_k \rangle \langle \psi_k |,
\end{equation}
where $|\psi \rangle $, $|\psi_k \rangle $ are eigenvectors and $\lambda$, $\lambda_k$ are non-negative
eigenvalues. We assume descending order throughout this work as $\lambda > \lambda_2 \geq \lambda_3 \dots$
and assume that the density matrix has a distinguished, unique dominant eigenvalue $\lambda$ (no degeneracy).

The core problem considered in the present work is the following: the dominant eigenvector $|\psi \rangle$ of the
noisy quantum state $\rho$ will be different from the ideal computational state $|\psi_{id} \rangle$
(and from eigenvectors of $\rho_{err}$), except in the special case when $\rho_{err}$ and $\rho_{id}$
commute. The reason is that in the commuting case the two density matrices share the same eigenvectors
and thus their sum will share the same eigenvectors too. However, in realistic physical systems $\rho_{err}$ and $\rho_{id}$
are highly unlikely to commute.
Surprisingly, even a completely
incoherent noise channel---such as depolarising and dephasing as described below Eq.~\eqref{error_channel}---can introduce a
coherent mismatch resulting in a coherently shifted dominant eigenvector as 
$|\psi \rangle = \sqrt{1-c}|\psi_{id} \rangle + \sqrt{c} |\psi_\perp \rangle$ in Eq.~\eqref{rho_decomp}. 
Our aim is to characterise and generally upper bound this coherent mismatch $c$.
Let us first formalise our definition of the coherent mismatch $c$ and then briefly motivate this work via important
scenarios where this coherent mismatch plays a crucial role. For example, the present problem is very closely
related to the well-known case of bounding the eigenvalues of a sum of two matrices which we discuss in Sec.~\ref{sec_related_maths}

\begin{definition}
	\label{def_state}
	 We define the coherent mismatch as the  infidelity between
	the dominant eigenvector $| \psi \rangle$ of a noisy quantum state $\rho$ from Eq.~\ref{rho_decomp} and the ideal
	computational state  $| \psi_{id} \rangle$ as
	\begin{equation} \label{coh_mismatch_def}
		c := 1 - |\langle \psi_{id} | \psi \rangle|^2.
	\end{equation}
	Here we also define the fidelity $F:= \langle \psi_{id} | \rho | \psi_{id} \rangle$.
	For some of the arguments later we will make use of the decomposition into a sum (for some $\eta>0$)
	\begin{align} \label{rho_decomp_sum}
		\rho 
		=&
		\eta \rho_{id} + (1-\eta)  \rho_{err}\\
		=& \eta |\psi_{id} \rangle \langle \psi_{id} | + (1-\eta)  \sum_{k=1}^d \mu_k |\chi_k \rangle \langle \chi_k |,
		\nonumber
	\end{align}
	of the ideal computational state $\rho_{id}:= |\psi_{id} \rangle \langle \psi_{id} |$ and a suitable
	error density matrix $\rho_{err}$
	(see text above). For this decomposition we can define the ratio of eigenvalues as
	$\delta := (\eta^{-1} - 1 ) \, \mu_1$,
	where $\mu_1$ is the largest eigenvalue of $\rho_{err}$.
\end{definition}
Notice that the the eigenvectors $|\chi_k \rangle$ above are generally different
than the ones in Eq.~\eqref{rho_decomp} (non-commuting case).
Let us remark that while the decomposition in Eq.~\eqref{rho_decomp_sum} is very useful
for illustrating and motivating the present problem, it is not necessary and some of the
later results in this work will be independent of this decomposition.
Refer to Appendix~\ref{def_state_appendix} for more details.

\subsection{Error suppression \label{sec_error_suppression}}

Two recent works \cite{koczor2020exponential, huggins2020virtual} have introduced an approach which
can suppress errors exponentially when preparing $n$ copies of a noisy quantum state -- and which was named
error suppression by derangements (ESD) and virtual distillation (VD).
The core idea behind the approach is that it prepares $n$ identical copies
	of a noisy computational quantum state $\rho$ and uses the copies to `verify each other' by
	applying a derangement operation (generalisation of the SWAP operation
	that permutes the $n$ registers). This filters out all error contributions
	that break global permutation symmetry among the copies, hence allows for exponential suppression
	when increasing $n$.

While ref.~\cite{huggins2020virtual} mostly focuses on the $n=2$ scenario and proposes
a resource efficient variant that does not require an ancilla qubit when $n=2$,
ref.~\cite{koczor2020exponential} presents explicit constructions of the approach for $n \geq 2$.
A possible implementation is illustrated in Fig.~\ref{plot_schem} which uses
a controlled-derangement operation and allows one to measure expectation values
	of the form $\tr[\rho^n \, O]/\tr[\rho^n]$ with respect to an observable $O$.
In this regard ref.~\cite{koczor2020exponential} notes that
for $n > 2$ a large number of possible derangement patterns exist while a qubit-efficient one
was proposed in the follow-up work \cite{czarnik2021qubit}.

When increasing the number of copies $n$, the `virtual' quantum state $\rho_n:= \rho^n/\tr[\rho^n]$
approaches the dominant eigenvector from Definition~\ref{def_state} in exponential order.
Since the dominant eigenvector $| \psi\rangle$ is generally
different from the ideal computational state $| \psi_{id} \rangle$ via Definition~\ref{def_state},
the coherent mismatch limits the ultimate precision of the ESD and VD approaches. Ref.~\cite{koczor2020exponential}
defined the coherent mismatch $c$ (see Definition~\ref{def_state}) to determine this discrepancy.

Similarly, the `noise floor' was defined in ref~\cite{huggins2020virtual} to express the discrepancy between 
the `virtual' quantum state $\rho_n$ and the ideal computational state
$\rho_{id}$ in the limit of a large number of copies via the trace distance $T(\rho_n, \rho_{id})$.
We prove in Appendix~\ref{app_noise_floor} that this noise floor is equivalent to the coherent mismatch
up to a square-root as
\begin{equation*}
	\lim_{n \rightarrow \infty} T(\rho_n, \rho_{id}) = \sqrt{c},
\end{equation*}
which confirms that indeed the notions of the coherent mismatch and noise floor are equivalent:
ultimately they both express the infidelity between the pure states $| \psi_{id} \rangle$ and $| \psi\rangle$.

%-------------------------
\begin{figure}[tb]
	\begin{centering}
		\includegraphics[width=0.40\textwidth]{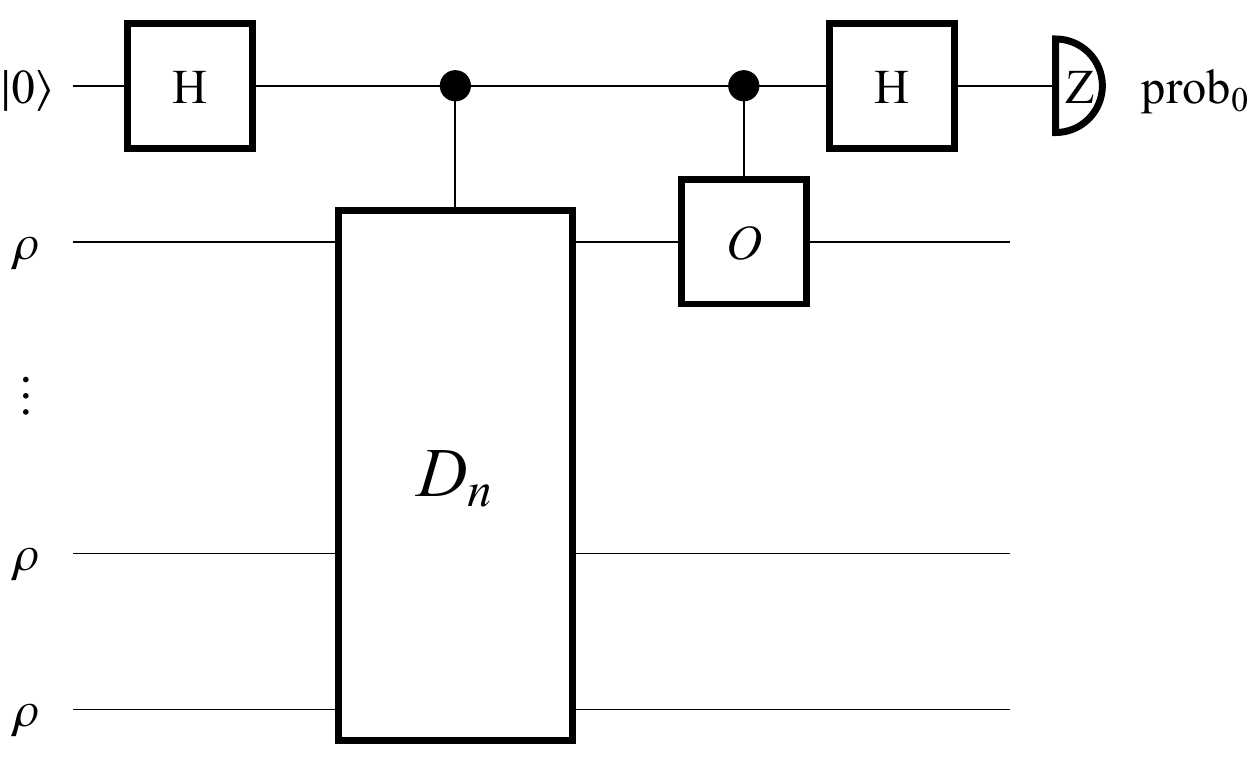}
		\caption{
		Quantum circuit of a possible implementation of the ESD/VD error
		suppression approach \cite{koczor2020exponential,huggins2020virtual} (figure adopted from~\cite{koczor2020exponential}).
		$n$ copies of the noisy quantum state $\rho$
		are prepared and entangled via a controlled-derangement operator $D_n$ -- a generalisation
		of the SWAP operation that permutes the $n$ quantum registers.
		The probability $\prob$ when measuring the ancilla qubit is proportional to 
		the expectation value of the observable $\tr[\rho^n O]$ in the `virtually
		distilled' states $\rho^n$. For large $n$, the approach ultimately allows to measure expectation
		values in the pure state $|\psi\rangle$ as the dominant eigenvector of $\rho$ from Eq.~\eqref{rho_decomp}.
		A qubit-efficient construction was proposed in \cite{czarnik2021qubit}.
		\label{plot_schem}
		}
	\end{centering}
\end{figure}
%-------------------------

Both works used a perturbative expansion of the dominant eigenvector $| \psi\rangle$ to approximate this infidelity.
While such perturbative series may be accurate in the limit of very low noise $\eta \rightarrow 1$ in Eq.~\eqref{rho_decomp_sum},
they are not applicable to the practically relevant scenario when quantum states accumulate a large amount
of noise. Furthermore, we establish in Remark~\ref{remark_perturb} that the perturbative series diverges in the worst-case scenario region.
It is thus the aim of the present work to
derive generally applicable upper bounds and approximations of the coherent mismatch that are generally applicable in
any scenario. As such, our bounds in Sec.~\ref{results} are saturated by extremal worst-case quantum states.
We will use these bounds to generally argue about the efficacy (number of copies, entropies etc.)
of the error suppression technique in complete generality in Sec.~\ref{results} -- which 
is beyond the scope of perturbation theory.

Ref.~\cite{huggins2020virtual} argued that the coherent mismatch is zero if the error channel
maps only to orthogonal states. Indeed such special density matrices are an instance of
the general class when $\rho_{err}$ and $\rho_{id}$ commute as discussed above. Interestingly,
we show that the worst case scenario quantum states, which maximise the coherent mismatch in
Theorem~\ref{theo_commut_upb}, have eigenvectors that are all orthogonal to the ideal state except for
the dominant error eigenvector. This highlights that, somewhat counter-intuitively,
the orthogonal error models proposed in ref~\cite{huggins2020virtual} produce quantum states
(with $c=0$) that are actually close in state space to the
worst-case quantum states (with almost all eigenvectors orthogonal to the ideal state)
that maximise $c$.

Ref.~\cite{koczor2020exponential} noted that the coherent mismatch is necessarily zero
when noise density matrices $\rho_{err}$ commute with the ideal state,
and gave the example of single qubit systems
undergoing depolarising noise. Ref.~\cite{huggins2020virtual} numerically simulated
this kind of scenario via non-entangling (random) circuits undergoing depolarising noise
and found that the noise floor is indeed zero.
Indeed, local depolarising noise in single-qubit systems maps to errors $\rho_{err} = \mathrm{Id}/d$
that commute with the ideal, unentangled state and one trivially finds
that $c=0$, regardless of whether the circuits are random or not. 
As such, ref.~\cite{huggins2020virtual} demonstrated that the noise floor $\sqrt{c}$ is indeed non-zero
and significant even for relatively deep, random entangling circuits. Results in Sec.~\ref{circuit_section}
can be applied to such random circuits and confirm the numerical observations that the coherent mismatch is
non-zero and decreases when increasing the depth of the circuit.

Ref.~\cite{koczor2020exponential} additionally observed numerical scaling results of the coherent mismatch
in terms of the number of gates and number of qubits in noisy quantum circuits. We confirm these scaling results
in Sec.~\ref{circuit_section} using general upper bounds.
Before stating the main results, let us first motivate the practical relevance of the present work.

\section{Motivation \label{sec_motivation}}

\subsection{Ultimate precision in error suppression \label{precision_motivation}}

The previously introduced ESD and VD approaches allow one to estimate the expectation value
$ \langle \psi | O |\psi \rangle$ for sufficiently large $n$. This expectation value can be
biased due to the coherent mismatch of the state $|\psi \rangle$ and will generally deviate
from the ideal expectation value $\langle \psi_{id} | O |\psi_{id} \rangle$. 

While we define and compute the coherent mismatch in terms of distance measures on the quantum states,
one can indeed relate it to the more practical question of how much error the discrepancy
between  $| \psi_{id} \rangle$ and $| \psi\rangle$ introduces into the measurement of an
observable $\langle \psi | O |\psi \rangle$. Ref.~\cite{huggins2020virtual} proposed that
the trace distance generally upper bounds these observable measurement errors as
\begin{align}
| \langle \psi_{id} | O |\psi_{id} \rangle - \langle \psi | O |\psi \rangle|  \nonumber
& \leq 
2 \lVert O \rVert_{\infty} \, T(|\psi \rangle \langle \psi | , |\psi_{id} \rangle \langle \psi_{id} | )  \\
& =  2  \sqrt{c}\, \lVert O \rVert_{\infty} ,
\label{eq_obs_error}
\end{align}
where $\lVert O \rVert_{\infty}$ is the absolute largest eigenvalue of the observable,
refer to Appendix~\ref{app_noise_floor} for a proof.
The second equality relates the trace distance to the coherent mismatch $c$. 

While this trace-distance measure is a general upper bound, it was already noted in ref.~\cite{huggins2020virtual}
that this bound is very pessimistic in practically relevant scenarios.
We demonstrate this in Fig.~\ref{plot_trdist} (blue): We randomly
generate $10^4$ quantum states and normalised observables (i.e., $\lVert O \rVert_\infty =1$)  for
randomly selected dimensions between $2 \leq d \leq 100$ and compute the actual error in the
observable measurements as $| \langle \psi_{id} | O |\psi_{id} \rangle - \langle \psi | O |\psi \rangle|$.
While in Fig.~\ref{plot_trdist} (blue) some random states get relatively close,
indeed, most of the randomly generated states are orders of magnitude below the upper bound.

To support the observation of ref.~\cite{huggins2020virtual} with a rigorous statement, we determine an
alternative bound in Appendix~\ref{app_noise_floor} for the specific but pivotal case when the aim is to prepare
eigenstates of the observable.
Note that the majority of quantum algorithms that target early quantum devices actually aim to prepare eigenstates
of certain Hamiltonian operators as $O \equiv \mathcal{H}$, see e.g., the review articles \cite{endo2020hybrid,cerezo2020variationalreview,bharti2021noisy}.
Remarkably, we show in Appendix~\ref{app_noise_floor} that if the quantum device prepares an eigenstate of the observable then the
error in estimating the ideal expectation value is upper bounded as
\begin{equation*}
	| \langle \psi_{id} | O |\psi_{id} \rangle - \langle \psi | O |\psi \rangle|
	\leq 
	2c \lVert O \rVert_{\infty},
	\quad \text{if}\quad
	O |\psi_{id} \rangle \propto |\psi_{id} \rangle,
\end{equation*}
which is a quadratically smaller (in $c$) bound than the one in Eq.~\eqref{eq_obs_error}.
We demonstrate in Fig.~\ref{plot_trdist} (orange) that the measurement errors in case of eigenstates
are indeed orders of magnitude below the pessimistic bounds (blue) and are generally upper bounded
by the orange line. Furthermore, in Fig.~\ref{plot_optimise} we illustrate that
in practical applications, such as the variational quantum eigensolver (VQE), even
approximate ground states produce errors significantly below the general bound.

The above bounds all depend on the actual value of $c$, and it is
thus the aim of the present work to comprehensively determine the coherent mismatch.

\subsection{Related problems in mathematics\label{sec_related_maths}}
Let us now relate the present work to important themes in mathematics.
In particular, it is a well-known problem in mathematics to generally bound eigenvalues of a sum
of two Hermitian matrices. The problem was first proposed by Weyl in 1912 \cite{weyl1912asymptotische}:
given two Hermitian matrices $A$ and $B$ with eigenvalues $\alpha_k$ and $\beta_k$, how does one determine
the eigenvalues $s_k$ of the sum of the two matrices $S=A+B$? Weyl's partial solution to this problem
determines the possible range that the eigenvalues of $S$ can take via the inequalities
\begin{equation*}
	s_{k+l-1} \leq \alpha_k + \beta_l, \quad \text{for indexes} \quad k+l -1 \leq d,
\end{equation*}
where $d$ is the dimension of the matrices and the eigenvalues are arranged in descending order.
A typical application of these inequalities is to
bound the possible eigenvalues of the sum as $s_k \leq  a_k + \beta_{\mathrm{max}}$
with $\beta_{\mathrm{max}} \equiv \beta_1$.
These partial results can be proven by minmax methods which can already be a considerable task.

Following a series of major breakthroughs in mathematics, this problem has only been solved relatively
recently to a full extent using honeycomb structures
\cite{klyachko1998stable,helmke1995eigenvalue,knutson1999honeycomb,knutson2004honeycomb}.
The final resolution specifies a set of inequalities in terms of the
eigenvalues $a_k,b_k, s_k$. We refer the interested reader to the excellent article \cite{knutson2001honeycombs}.

This highlights the complex and difficult nature of predicting the eigensystem of the sum of two matrices.
While bounds on eigenvalues have been completely solved by the application of the honeycomb structures,
much less is known about the eigenvectors of the sum of two matrices. It is the aim of the present work to
determine general bounds on the dominant eigenvector of the sum of two matrices as introduced
in Definition~\ref{def_state}.

The current problem is, however, special: while we do not make any assumption about the
matrix $\rho_{err}$, our matrix $\rho_{id}$ is a rank-1 projector and thus its eigenvalues are
$a_k = 0$ for all $k\geq 2$. Due to this special structure, Weyl's inequalities
are significantly simplified in the present scenario, and this allows us to obtain the
	following straightforward bounds.
\begin{remark}\label{remark_weyl}
	Straightforwardly applying Weyl's inequalities generally guarantees
	that $\lambda_2 < \lambda $ and thus the dominant
	eigenvector corresponds to  $| \psi_{id} \rangle$ as long as $\delta < 1$ due to the following bounds.
	In particular, applying Weyl's inequalities to Definition~\ref{def_state} suffices to generally upper bound
	the two largest eigenvalues $\lambda$ and $\lambda_2$ of the noisy density matrix from Eq.~\eqref{rho_decomp}
	(or similarly any other eigenvalues) as
	\begin{equation*}
		\eta	\leq \lambda \leq \eta (1+ \delta)
		\quad \text{and} \quad  0 \leq \lambda_2 \leq \eta \delta.
	\end{equation*}
	Here $\eta$ and $\delta$ were defined in Definition~\ref{def_state}.
\end{remark}

Although this work considers relatively special matrices, it is a considerable task to go beyond eigenvalues and to
determine eigenvectors of a sum of two matrices, i.e., as relevant for the coherent
mismatch. Let us highlight how the present problem crucially deviates from the previously discussed
case of eigenvalues.

The Weyl inequality in the above remark is saturated when the two matrices have the same 
dominant eigenvectors leading to an extremal shift in the dominant eigenvalue.
This however implies that $\rho_{err}$ and $\rho_{id}$ commute thus leading to a coherent mismatch
that is zero, i.e., no shift in the dominant eigenvector. On the other hand, in Sec.~\ref{sec_upbdelta}
we determine extremal states that maximise the coherent mismatch and their structure is indeed
in stark contrast to the case of the eigenvalues.

It is worth noting that the present work makes connections to and uses results from other
topics in mathematics: analytical results are used for computing eigenvalues and eigenvectors
of arrowhead matrices in Sec.~\ref{sec_arrowhead}
and new bounds are established in Sec.~\ref{sec_comm_norm} for the matrix norm of commutators
-- this improves upon known general results in the considered specific scenarios.
Let us now derive our results.

\section{Results\label{results}}

\subsection{General upper bounds and extremal states}

Let us first exploit that the present work 
considers a relatively special structure since the matrix $\rho_{id}$ is a rank-1 projector:
we now introduce a special decomposition of the matrix $\rho$ which will allow us to compute $c$
analytically and thus to  construct extremal, worst-case scenario
quantum states, i.e., families of quantum states that are guaranteed to saturate 
upper bounds on $c$.

\subsubsection{Arrowhead matrices\label{sec_arrowhead}}

\begin{statement}\label{stat_arrowhead}
	The quantum state $\rho$ in Definition~\ref{def_state} is unitarily equivalent to a \textbf{real, symmetric, non-negative} arrowhead
	matrix and can be decomposed into the sum of matrices $\tilde{\rho} = F |\tilde{\psi}_{id}\rangle\langle \tilde{\psi}_{id}| + D+C$ as
	\begin{equation}
		\tilde{\rho}  = \begin{pmatrix}
			F & C_2 & C_3 & \dots & C_{d}\\
			C_2 & D_2 &  &  &  \\
			C_3 &  & D_3 \\
			\vdots & & & \ddots &\vdots \\
			C_{d} &  &  & \dots & D_{d}
		\end{pmatrix}.
	\end{equation}
	We have applied a unitary transformation $\tilde{\rho} := U \rho U^\dagger$
	such that $|\tilde{\psi}_{id}\rangle := U |\psi_{id}\rangle = (1,0, \dots 0)$ 
	while $F,C_k,D_k \geq 0$ with $k\in\{2, 3, \dots,d \}$ with $d$ denoting the dimension,
	and all other matrix entries are zero.
\end{statement}

Refer to Appendix~\ref{proof_arrowhead} for a proof.
These so-called arrowhead matrices have unique properties and  have been investigated in the literature extensively.
For example, certain matrix algorithms use arrowhead matrices to speed up computations  \cite{gu1995divide}
and further applications include, e.g.,
 the description of radiationless transitions in isolated molecules \cite{bixon1968intramolecular}
or of oscillators vibrationally coupled with a Fermi liquid \cite{gadzuk1981localized}.
Let us mention two remarkable properties of these special matrices.

First, Cauchy's interlacing theorem guarantees that the entries $D_k$ satisfy the
general interlacing inequalities with the eigenvalues $\lambda_k$ from Eq~\eqref{rho_decomp} as
\begin{equation} \label{eq_interlacing}
	 D_2 \geq \lambda_2 \geq D_3 \geq \lambda_3  \geq \dots D_d \geq \lambda_d,
\end{equation}
refer to, e.g., ref.~\cite{o1990computing} for more details.

Second, if one knows the explicit representation of the arrowhead matrix, i.e., knowing the matrix entries $D_k$, $C_k$
and $F$, then the eigenvalues $\lambda$ and $\lambda_k$ can be obtained as roots of the secular function \cite{o1990computing}
\begin{equation} \label{eq_eigenvalue} 
	P(x) = x - F + \sum_{k=2}^d \frac{C_k^2 }{(D_k - x) }.
\end{equation}

%-------------------------
\begin{figure*}[tb]
	\begin{centering}
		\includegraphics[width=0.85\textwidth]{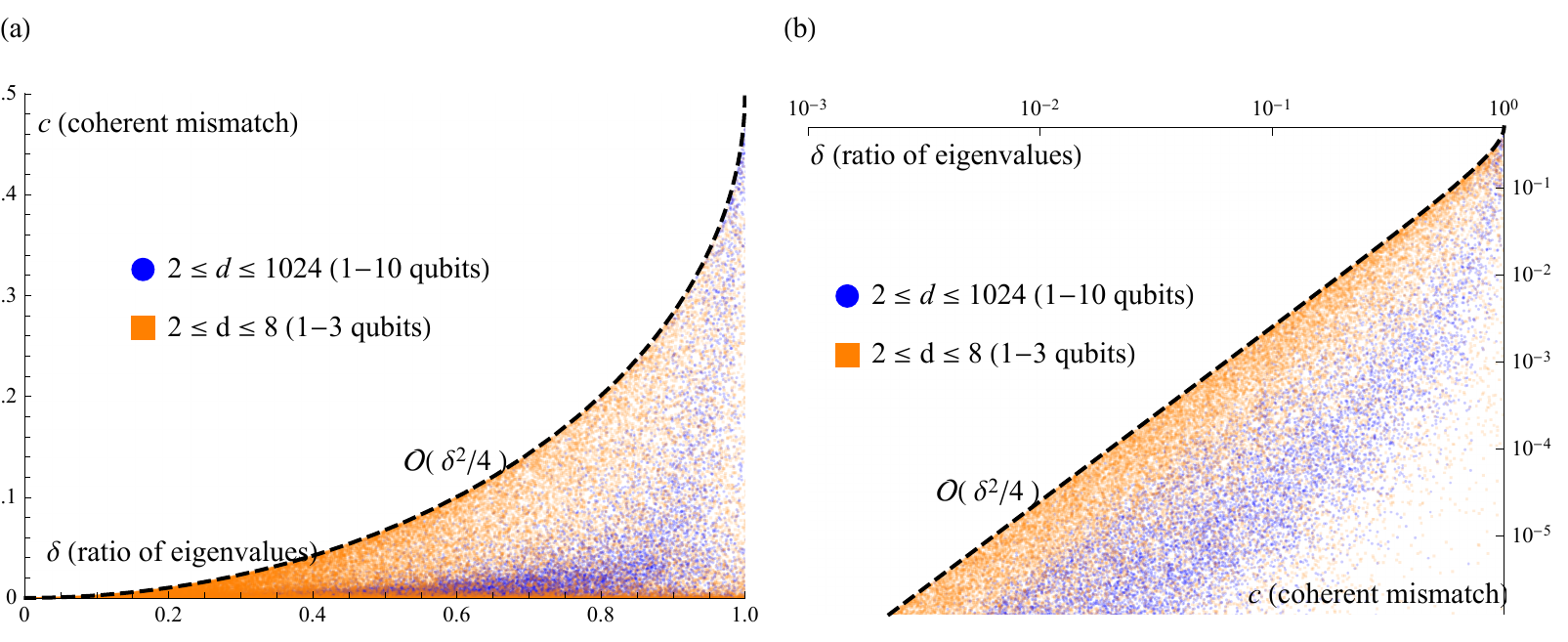}
		\caption{
			Coherent mismatch $c$ in randomly generated states and its upper bound (dashed black lines) as a function
			of the ratio $\delta$ of the largest error eigenvalue vs. the ideal state's contribution $\eta$ from
			Definition~\ref{def_state}. 
			(a) linear-linear scale and (b) log-log scale.
			The dominant eigenvector $| \psi \rangle$  of a noisy quantum state $\rho = \eta |\psi_{id} \rangle \langle \psi_{id} | + (1-\eta)  \rho_{err}$
			is generally different than the ideal computational state as characterised by the coherent mismatch
			(infidelity) $c= 1 - |\langle \psi_{id} | \psi \rangle|^2$.
			The general upper bound on $c$ from Theorem~\ref{theo_upbdelta} (dashed black line) is saturated by the extremal quantum states
			$\rho_{err}$  which are highly unlikely to appear in practical scenarios and thus experimental quantum states
			are expected to be significantly below this bound.
			Here, $\delta$, $\rho_{err}$ and $c$ are defined in Definition~\ref{def_state}.
			Randomly (uniformly with respect to the Haar measure) generated quantum states of large dimensions (blue dots) are
			significantly less likely to saturate the bounds than quantum states in smaller dimensions
			(orange rectangles). We present better lower and upper bounds in Sections~\ref{sec_lower_upper_bounds},
			see also Fig.~\ref{plot_commutators}.
			\label{plot_eigvals}
		}
	\end{centering}
\end{figure*}
%-------------------------

\subsubsection{Analytically solving the coherent mismatch \label{sec_arrowhead_eigenvector}}
The most important consequence of the previously introduced arrowhead structure is that,
given the knowledge of the decomposition of the density matrix $\rho$ into
the arrowhead form, we can analytically solve its eigenvectors and obtain an analytical expression
for the coherent mismatch.
\begin{statement}\label{stat_coherent}
	We can analytically compute the coherent mismatch in terms of the dominant eigenvalue $\lambda$ from Eq.~\eqref{rho_decomp}
	and in terms of the arrowhead matrix entries $D_k,C_k$ from from Statement~\ref{stat_arrowhead} as
	\begin{equation}
	c := 1- |\langle \psi_{id} | \psi \rangle|^2 
	= 1- [1 + \sum_{k=2}^d \frac{C_k^2 }{(\lambda - D_k)^2 } ]^{-1}.
\end{equation}
\end{statement}
Refer to Appendix~\ref{proof_coherent} for a proof.
The above formula allows us to analytically compute the coherent mismatch if the arrowhead form
of the density matrix is known. Even though we do not necessarily know such a decomposition
explicitly for arbitrary quantum states $\rho$, the
above formula is a very important ingredient for our following derivations and allows us to
derive general upper and lower bounds on the coherent mismatch. Before stating these results,
let us briefly remark on the striking resemblance of the above equation to perturbation theory. 
\begin{remark}\label{remark_perturb}
	Using first-order perturbation in order to approximate the dominant eigenvector
	(refer to, e.g., Eq.~(5.1.44) in \cite{sakurai1995modern} and to Eq.~(10.2) in \cite{wilkinson1965algebraic})
	enables us 	to estimate the coherent mismatch as
	\begin{equation*}
	c_{pert}^{(1)} = 1 - [1 + \sum_{k=2}^{d} \frac{C_k^2}{(F-D_k)^2}]^{-1}.
\end{equation*}
	This approximation is formally similar to the exact analytical formula
	of the coherent mismatch from Statement~\ref{stat_coherent}, but note that
	here we need to divide with
	the factor $(F-D_k)^2$ and not with $(\lambda- D_k)^2$. This approximation breaks down
	in the region when quantum states accumulate a large amount of noise and $F \approx D_k$.
\end{remark}
Refer to Appendix~\ref{proof_perturb} for a proof.
It is interesting to note the connection to first-order perturbation theory , which also
confirms that, indeed, the above expression is accurate when the noise in the state
(via $\eta \rightarrow 1$ in Eq.~\eqref{rho_decomp_sum}) is vanishingly small and thus we obtain $F \approx \lambda$
with $F \gg D_2$ .

\subsubsection{Upper bound via extremal states\label{sec_upbdelta}}
We will now use the above introduced arrowhead decomposition of density matrices and derive a family of quantum states
that maximise the coherent mismatch. We analytically solve this optimisation problem in
Appendix~\ref{proof_upbdelta} and find that the maximum of the coherent mismatch is attained only by the
following extremal density matrices: In the arrowhead representation of these states the only non-zero 
off-diagonal component is given by  $C_2$ (all other off-diagonal components are zero as $C_k=0$ for $k>2$)
while the diagonal entries $F$ and $D_k$ can be arbitrary.

Due to this simplified structure, we can analytically compute the coherent mismatch
which then serves as a general upper bound.
\begin{theorem}\label{theo_upbdelta}
	The coherent mismatch is generally upper bounded as	
	\begin{equation*}
		c \leq (1-\sqrt{1-\delta^2} )/2 = \delta^2/4 + \mathcal{O}(\delta^4/16),
	\end{equation*}
	where $\delta$ was defined in Definition~\ref{def_state}. This
	upper bound is saturated by an infinite number of worst-case error density matrices $\rho_{err}$
	whose dominant eigenvector $| \chi \rangle $ has a non-zero overlap with the ideal state
	$| \psi_{id} \rangle$ as 
	\begin{equation*}
		| \chi \rangle := \sqrt{\alpha}| \psi_{id} \rangle +  \sqrt{1-\alpha}| \phi_2 \rangle,
	\end{equation*}
	and all other eigenvectors of $\rho_{err} $ are orthogonal to the ideal state $| \psi_{id} \rangle$.
	The coherent mismatch is maximised when $\alpha = (1 + \delta)/2$ and note that the two basis
	vectors are orthogonal $\langle \phi_2 |  \psi_{id} \rangle = 0$.
\end{theorem}

The above theorem establishes that the worst kind of error density matrices
$\rho_{err}$ are the ones in which only the dominant eigenvector has a non-zero
overlap with the ideal state $\rho_{id}$ while all other eigenvectors are orthogonal to
the ideal state. Only these kind of errors can saturate the general upper bound on $c$,
however in stark contrast, quantum circuits in near-term quantum devices typically produce
error density matrices whose eigenvectors are highly unlikely to be orthogonal to the ideal state.
It thus stands to reason that the extremal error density matrices are highly unlikely to appear in practice,
and thus practically relevant noisy quantum states are expected to be significantly below this bound.

An important implication of the above theorem for practical applications is that the error bound depends
on the dominant eigenvalue $\mu_1$ of the noise state $\rho_{err}$ (since $\delta$ is proportional to $\mu_1$).
This eigenvalue depends exponentially on the Rényi
entropy $\mu_1 = e^{- H_\infty}$ which generally lower bounds all other Rényi entropies as  $H_\infty \leq \dots  H_2 \leq H_1$. 
We are thus guaranteed that the coherent mismatch decreases exponentially with Rényi entropies of the error density matrix
eigenvalues.
Similar exponential scaling results were obtained in ref.~\cite{koczor2020exponential} for the ESD approach and
it was noted that near quantum hardware are be expected to produce large entropy quantum states.
As such, a significant advantage of the present upper bound is that the parameter $\delta$ depends only
on spectral properties of the quantum state, i.e., eigenvalues and Rényi entropies, which may be
estimated in experiments~\cite{ekert2002direct, PhysRevA.64.052311, marvian2014generalization, 9163139, PhysRevA.89.012117,christandl2007nonzero,christandl2006spectra}.

Fig.~\ref{plot_eigvals} shows the coherent mismatch in case of $5\times10^4$ randomly
generated quantum states. Orange rectangles (blue dots) in Fig.~\ref{plot_eigvals}  correspond to quantum states whose dimension $d$
was generated uniformly randomly in the range $2 \leq  d \leq 8$ ($2 \leq  d \leq 1024$). Indeed, saturating
the upper bound (dashed black line) is significantly less likely in larger dimensions (blue rectangles
are significantly below upper bound).
This is expected since the extremal quantum states occupy a rapidly decreasing portion
of the full volume of state space.
Refer to Appendix~\ref{app_numerics} for more details.

\subsubsection{Limiting scenarios\label{sec_limiting_sc}}
We have found in the previous section that the error states $\rho_{err}$ that saturate the error bound
depend on the parameter $\delta$ which quantifies the ratio of the eigenvalues
(ideal state vs. dominant eigenvalue of the error state $\rho_{err}$, see Definition~\ref{def_state}).

In the limiting scenario when the contribution of the error density matrix $\rho_{err}$
is much smaller than the ideal state we obtain the limit $\delta \rightarrow 0$.
In this limit the dominant eigenvector of the extremal error state $\rho_{err}$ is an equal superposition
\begin{equation}\label{eq_extremal_state}
	| \chi \rangle = ( | \psi_{id} \rangle +  | \phi_2 \rangle )/\sqrt{2}, 
\end{equation}
due to Theorem~\ref{theo_upbdelta} where $| \phi_2 \rangle$ is an arbitrary error state that is orthogonal to
$| \psi_{id} \rangle$. This also informs us 
that the extremal quantum states in the practically relevant regime (i.e., for small $\delta$) have
dominant error vectors of the form, i.e., $| \chi \rangle \approx ( | \psi_{id} \rangle +  | \phi_2 \rangle )/\sqrt{2}$.

On the other hand, when the contribution of the error state is as strong as the ideal state via $\delta \rightarrow 1$
then the worst-case error vector is almost orthogonal to the ideal state 
via some small $\omega \ll 1$
\begin{equation*}\label{eq_extremal_state_global}
	| \chi_\omega \rangle = \sqrt{\omega}  | \psi_{id} \rangle +  \sqrt{1-\omega} | \phi_2 \rangle.
\end{equation*}
Surprisingly, we find that the global worst-case error, i.e., when $c=1/2$, can only be saturated by the quantum state
in the limit when $\omega \rightarrow 0$ (one must compute the limit only after computing $c$) as
$$ 
\rho = \frac{1}{2} | \psi_{id} \rangle \langle  \psi_{id} |  + \frac{1}{2} | \chi_\omega \rangle \langle  \chi_\omega |
=
\tfrac{1}{2}
\left(
\begin{array}{cc}
	1 &  \sqrt{\omega} \\
	 \sqrt{\omega} & 1 \\
\end{array}
\right)	
+\mathcal{O}(\omega).
$$
To illustrate this, in the second equation above we have computed the matrix representation
of the quantum state in the $2$-dimensional subspace spanned by
the orthonormal vectors $ | \psi_{id} \rangle $ and $| \phi_2 \rangle$
to leading order in $\omega$. Indeed, the dominant eigenvector of
this density matrix is the vector $(1,1)^T/\sqrt{2}$
(up to an error $\mathcal{O}(\omega)$) and this vector has a fidelity $1/2 +\mathcal{O}(\omega)$ to the ideal
computational state $(1,0)^T$. The limit of the coherent mismatch $\lim_{\omega \rightarrow 0} c = 1/2$
is thus well-defined, however, note that
the state itself in the limit becomes trivially the identity matrix (commuting case).

Interestingly, here we find exactly the opposite behaviour when compared to the case of
eigenvalues in Weyl's inequalities in Sec.~\ref{sec_related_maths}. Recall that
for the sum of two matrices $\rho = (| \psi_{id} \rangle \langle  \psi_{id} | + \rho_{err})/2$ the extremal shift to the eigenvalues (Weyl inequalities)
is saturated when the dominant eigenvector of $\rho_{err}$ is actually $| \psi_{id} \rangle$.
In stark contrast, we have found above that the extremal coherent mismatch (extremal shift in the dominant eigenvector)
is saturated only in the limit when the dominant eigenvector of $\rho_{err}$ is orthogonal to $| \psi_{id} \rangle$.

%-------------------------
\begin{figure*}[tb]
	\begin{centering}
		\includegraphics[width=0.85\textwidth]{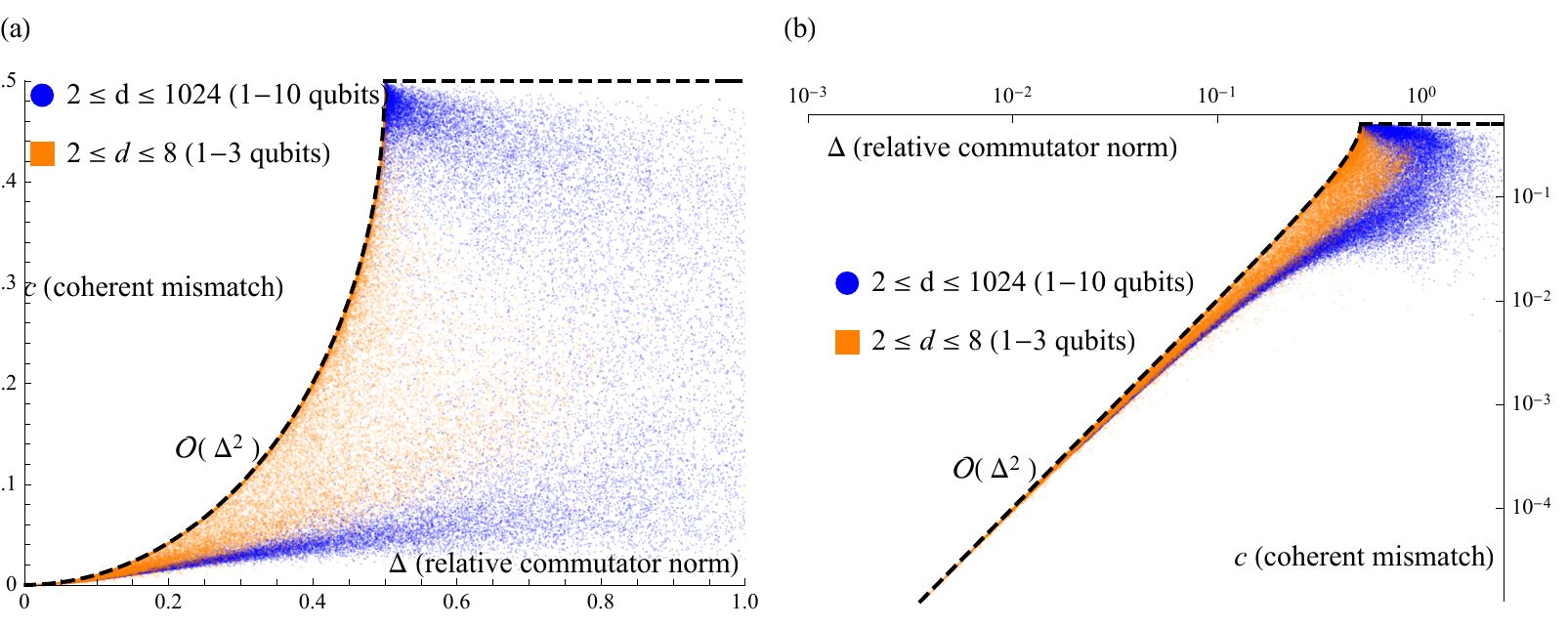}
		\caption{
			Coherent mismatch $c$ in randomly generated states and its upper bound (dashed black lines) as 
			a function of the relative commutator norm $\Delta$ (that is proportional to $\lVert [\rho_{id}, \rho] \rVert_\infty$).
			(a) linear-linear scale and (b) log-log scale.
			This bound is independent of the not necessary
			unique decomposition in Eq.~\eqref{rho_decomp_sum}.
			Another significant advantage is that this upper bound comes with a similarly scaling lower bound:
			for small $c\leq 10^{-3}$ all randomly (uniformly with respect to the Haar measure) generated states (blue dots and orange rectangles) nearly saturate
			the upper bound (dashed black lines) due to the asymptomatically coinciding lower and upper bounds.
			\label{plot_commutators}
		}
	\end{centering}
\end{figure*}
%-------------------------

\subsubsection{Application to error suppression \label{numb_copies}}
Let us finally remark on the implications of the above results to the performance of the ESD and VD approach.
Recall that ref.~\cite{koczor2020exponential} established general scaling results on how many
copies $n$ are required to reach a precision $\mathcal{E}$ in suppressing the noise
when measuring expectation values in the dominant eigenvector.

Let us now assume that the aim is to suppress this error level $\mathcal{E}$ 
to the level of error caused by the coherent mismatch (assuming normalised observables $\lVert O \rVert_{\infty} =1$).
Consistent with Theorem~\ref{theo_upbdelta}, we assume that the quantum state and the noise is of the form
of Eq.~\eqref{rho_decomp_sum} as $\eta \rho_{id} + (1-\eta)  \rho_{err}$ and assume that the quantum 
states are considerably noisy with $1-\eta$ being sufficiently
large as relevant in practice, i.e., $\eta \leq 2/3$ in general and $\eta \leq 4/5$ when we aim to
prepare eigenstates. These two conditions correspond to circuit error rates $\xi > 0.41$ and $\xi > 0.22$, respectively,
which is reasonable to assume in practice.
If we set our target precision to be the general trace distance bound from
Sec.~\ref{precision_motivation} as $\mathcal{E} = 2\sqrt{c}$,
then we obtain the following result in Appendix~\ref{proof_copies}:
we find that we need at least $3$ copies to reach the target precision with the worst-case extremal states.
Interestingly, ref.~\cite{czarnik2021qubit} found in numerical simulations of noisy derangement
circuits that, for the considered circuits, at least $3$ copies were required to reach a noise floor
determined by the coherent mismatch and by the noise in the controlled-SWAP operations.

On the other hand, if our aim is to prepare eigenstates as discussed in Sec.~\ref{precision_motivation},
then the coherent mismatch is guaranteed to cause a quadratically smaller error. 
We thus set the target precision to $\mathcal{E} = 2 c$ and find in Appendix~\ref{proof_copies}
that we need at least $4$ copies to reach the noise floor in the practically relevant region
where states are considerably noisy (via circuit error rates $\xi > 0.22$).
This confirms prior numerical simulations (Fig.~4 in ref.~\cite{koczor2020exponential}).

While we have derived these results for the extremal quantum states,
it stands to reason that in more realistic scenarios one may need significantly more copies to
reach the precision as limited by the coherent mismatch. 
Furthermore, these arguments establish that as long as the quantum device is limited to preparing only a small number of copies
(e.g., $2$, $3$ or $4$ depending on hardware constraints) of the noisy quantum state, then the error introduced by the coherent mismatch
will be guaranteed to be smaller than the error caused by having too few copies (not sufficient suppression).

\subsection{Lower and upper bounds via commutators\label{sec_lower_upper_bounds}}

While the previously derived bounds are tight as they are
saturated by the extremal states, they can be generally very pessimistic since
the extremal states are very unlikely to be relevant in practice.
This is nicely illustrated in Fig.~\ref{plot_eigvals} where most randomly generated
quantum states are significantly below this bound (dashed black lines) especially as the dimensionality grows
(orange vs. blue). In fact, the previous bound can be arbitrarily pessimistic, since
generally there is no lower bound of $c$ in terms of $\delta$: when $\delta$
is non-zero then $c$ can still be zero when $\rho_{err}$ and $\rho_{id}$ commute.
This leads to our next point: to derive general upper and lower bounds in terms of the
commutator. These bounds will in turn be independent of the non-unique decomposition in
Eq.~\eqref{rho_decomp_sum} and will also allow us to derive scaling results due
to the asymptotically coinciding lower and upper bounds.

\subsubsection{Expressing the commutator norm \label{sec_comm_norm}}
As we discussed above, if the error matrix $\rho_{err}$ commutes with the ideal state $\rho_{id}$
than the coherent mismatch must vanish. Similarly we would expect that
if the commutator is `large' than the coherent mismatch should also be large.
In the following we would like to introduce a measure of how large the commutator is. For this
purpose we will use a suitable matrix norm $\lVert \cdot \rVert$ which we will aim to upper bound.

Interestingly, it has been an open problem in mathematics to upper bound the Hilbert-Schmidt or Frobenius
norm of the commutator between two matrices and was only very recently solved for general matrices, refer to
refs.~\cite{bottcher2008frobenius,vong2008proof,wu2010short,bottcher2005big,laszlo2007proof,cheng2010commutators,wenzel2010impressions}
for more details. In particular, it was found that the norm of the commutator of two
generic matrices is upper bounded as
\begin{equation} \label{general_commut_norm}
	\lVert [A,B] \rVert_{HS} \leq \sqrt{2}	\lVert A \rVert_{HS} \, \lVert B \rVert_{HS}.
\end{equation}

As opposed to generic matrices, in the present case we aim to express the norm of the commutator
of two density matrices $[\rho_{id}, \rho]$. Although we make no assumption about $\rho$ (except that it is a density matrix),
$\rho_{id}$ is a special matrix,
i.e., a projector, since it represents a pure quantum state. This property allows us to express
the commutator norm more explicitly.
\begin{statement}\label{stat_commutator}
	We analytically solve the eigenvalues and eigenvectors of both the matrix $C$ from Statement~\ref{stat_arrowhead}
	and the commutator $[\rho_{id}, \rho]$. We establish that both matrices
	have only two non-zero eigenvalues as  $\mathrm{Spec}(C) = \{ \pm \sigma \}$ and
	$\mathrm{Spec}([\rho_{id}, \rho]) =  \{ \pm i \sigma \}$. It follows that
	their matrix norms are equivalent
	\begin{equation*}
		\lVert [\rho_{id}, \rho] \rVert_p = \lVert C \rVert_{p}   = 2^{1/p} \sigma,
	\end{equation*}
	for all $0 \leq p\leq \infty$.
	The eigenvalue can be computed as
	\begin{equation*}
		\sigma^2 = \var[\rho] =   \langle \rho^2 \rangle -  \langle \rho \rangle ^2 = \langle  \psi_{id} | \rho^2 | \psi_{id} \rangle  - F^2,
	\end{equation*}
	which expresses a generalised variance of the density matrix and $F$ is the fidelity.
\end{statement}
Refer to Appendix~\ref{proof_commutator} for a proof.
Interestingly, we can directly relate the off-diagonal entries of the arrowhead matrix---as determined by $C$---from Statement~\ref{stat_arrowhead}
to the commutator $[\rho_{id}, \rho]$. Furthermore, the
above result establishes that the commutator norm is exactly given by a
generalised uncertainty $\var[\rho]$ which is a notion widely used in quantum theory to
express the variance of measurement statistics of an observable in, e.g., quantum metrology \cite{review}
and beyond \cite{sakurai1995modern}. In the present case the observable is the operator $\rho$
and the state is the ideal computational state $| \psi_{id} \rangle$. It is also interesting to note that
this commutator norm $\sigma^2$ is proportional to the quantum Fisher information~\cite{qfi1} of the quantum state
$| \psi_{id} \rangle$ in a unitary parametrisation generated by the Hamiltonian $\mathcal{H} \equiv \rho$.

Let us now illustrate how using the above expressions yield improved bounds when
compared to the general bounds considered in the literature. As such, it is straightforward to
show that
\begin{equation*}
\lVert [\rho_{id}, \rho] \rVert_{HS} =  \sqrt{2} \sigma \leq \sqrt{2} \, \lambda,
\end{equation*}
and this bound is indeed considerably tighter than the prior general result
in Eq.~\ref{general_commut_norm} since for states of low purity (i.e., $\tr\rho^2 \ll 1$)
we find that $\lambda \ll \lVert  \rho \rVert_{HS}$.
Furthermore, assuming the decomposition from Eq.~\eqref{rho_decomp_sum} we obtain the general bound
$\sigma \leq \eta \delta/2$, where $\delta$ was defined in Definition~\ref{def_state}
and $\eta \delta$ was the extremal shift in the Weyl inequalities in Remark~\ref{remark_weyl}.

\subsubsection{Upper bound via commutator norm \label{sec_commut_upb}}
We are now prepared to derive a general upper bound of the coherent mismatch
based on the previously obtained norms of the commutator.

\begin{theorem}\label{theo_commut_upb}
	Let us define the metric $\Delta:=\sigma_r/(1- Q)$ that depends only on two parameters:
	the relative commutator norm $\sigma_r := \sigma/\lambda$, where the commutator norm $\sigma$ was defined in
	Statement~\ref{stat_commutator} and
	the ratio of the two dominant eigenvalues is $Q:=\lambda_2/\lambda$.  For any fixed $\Delta$ there exist an infinite number of
	worst-case scenario states that saturate the upper bound of the coherent mismatch as
	\begin{equation}\label{eq_commut_upb_theo}
		c \leq  (1- \sqrt{1- 4\Delta^2})/2 = \Delta^2	+ \Delta^4 +\mathcal{O}(\Delta^6).
\end{equation}
	These extremal states $\rho$ have eigenvectors $| \psi_k \rangle$ that are orthogonal
	to the ideal state $| \psi_{id} \rangle$,
	except for the two dominant eigenvectors $| \psi \rangle,| \psi_2 \rangle$ 
	that correspond to the two dominant eigenvalues $\lambda, \lambda_2$.
\end{theorem}

The above upper bound is saturated by extremal states similar to the ones in Theorem~\ref{theo_upbdelta}.
The crucial difference, however, is that this upper bound is completely independent of the (not necessarily
unique) decomposition  into an ideal and noisy quantum states from Eq.~\eqref{rho_decomp_sum}. The present bound
can thus be applied to more general scenarios too (note that the definition of the extremal states above is
independent of $\rho_{err}$).

Fig.~\ref{plot_commutators} shows the coherent mismatch as a function of the metric $\Delta$
for $5 \times 10^4$ randomly generated density matrices in various dimensions (blue dots and orange rectangles).
The upper bound (dashed black lines) is significantly more likely to be saturated by random states in lower
dimensions (orange rectangles) since the extremal states occupy a negligible volume of the increasingly higher
dimensional state space. We can identify two distinct regions in the plots.

First, for large $\Delta \geq 0.2$ most of the randomly generated states are significantly below the bound, similarly
as in Fig.~\ref{plot_eigvals}. Note also that the metric $\Delta$ can in principle be larger than $1/2$
and in such a scenario Eq.~\eqref{eq_commut_upb_theo} is not defined. For this reason Fig.~\ref{plot_commutators}
shows the general bound $c \leq 1/2$ in this region. We note, however, that this region is not relevant in practice
since typical quantum circuits in near-term quantum devices produce errors that typically result in relatively small
commutator norms $\Delta \ll 1$ as discussed in Sec.~\ref{circuit_section}.

Second, in the practically more relevant region where $c \leq 10^{-3}$ is sufficiently small,  one can
observe that all the randomly generated states nearly saturate the upper bound. 
The reason for this behaviour will be clarified in the next section where we derive
a general lower bound on $c$ and show that it approaches the upper bound as $c$ decreases
-- thus tightly confining the possible values that $c$ can take up. 
Let us now introduce this lower bound.

\subsubsection{Lower bound via commutator norm and application to error suppression\label{sec_lower_bound}}

Using the same technique as in Theorem~\ref{theo_commut_upb}
we can derive a directly analogous lower bound for the coherent mismatch.
\begin{lemma}\label{theo_commut_lowb}
	Let us define the metric $\dmin:=\sigma_r/(1- \qmin)$ that depends only on two parameters:
	the relative commutator norm $\sigma_r := \sigma/\lambda$ and the ratio $\qmin:=\lambda_m/\lambda$
	where $\lambda_m$ is the smallest non-zero eigenvalue of $\rho$. 
	For any fixed $\dmin$ there exist an infinite number of
	best-case scenario states that saturate the lower bound of the coherent mismatch as
	\begin{equation}
		c \geq  (1- \sqrt{1- 4 \dmin^2})/2 = \dmin^2	+ \dmin^4 +\mathcal{O}(\dmin^6).
	\end{equation}
    The dominant eigenvector $| \psi \rangle$ of the extremal state $\rho$ and its eigenvector $| \psi_m \rangle$
	that corresponds to the smallest non-zero eigenvalue $\lambda_m$ have non-zero overlaps 
	with the ideal computational state $|\psi_{id}\rangle$. All other eigenvectors $| \psi_k \rangle$
	of $\rho$ 
	with $k \in \{2, 3, \dots m-1, m+1, \dots d \}$	are orthogonal to the the ideal state $|\psi_{id}\rangle$.
\end{lemma}

Refer to Appendix~\ref{proof_lowb_commutator} for a proof.
The above Lemma guarantees that the coherent mismatch is always at least as large as
the above lower bound for a fixed $\dmin$. Note that  the upper bounds in Theorem~\ref{theo_commut_upb}
are similarly determined by $\sigma_{r}^2$: the most important consequence is that
for a sufficiently small coherent mismatch $c \rightarrow 0$, the possible values
that $c$ can take up are tightly confined by the upper and lower bounds. This is illustrated in Fig.~\ref{plot_commutators}:
all randomly generated states with small $c \leq 10^{-3}$ nearly saturate the upper bound.

To substantiate this observation let us compute the ratio of the lower and upper bounds as
\begin{equation}\label{eq_ratio}
	\frac{\text{lower b.}}{\text{upper b.}} = \frac{1-Q}{1-\qmin} + \mathcal{O}(\sigma^4)
	\approx 1 - \frac{\lambda_2-\lambda_m}{\lambda}.
\end{equation}
Let us now consider 3 different scenarios in which the above ratio approaches $1$ and thus
the lower and upper bounds coincide.

First, the ratio approaches $1$ when the suppression factor is very small $Q \ll 1$.
Such a small suppression factor guarantees high efficacy of the ESD/VD approach as established
in ref.~\cite{koczor2020exponential}, but it may not be reasonable to expect vanishingly small
suppression factors for realistic noisy circuits with a large number of gates, refer to Sec.~\ref{circuit_section}.
On the other hand, even a realistic $Q \approx 1/2$ would result in approximately a factor
of $2$ ratio between the lower and upper bounds which is already reasonably tight.

Second, the approximation in Eq.~\ref{eq_ratio} depends on the difference between the largest $\lambda_2$
and smallest $\lambda_m$ `error' probabilities (eigenvalues of $\rho$ from Eq.~\eqref{rho_decomp}).
Indeed, $Q$ need not vanish in order for the ratio in Eq.~\eqref{eq_ratio}
to approach $1$: it is sufficient that the smallest and largest `error' probabilities are close via
$\lambda_2 \approx \lambda_m$. 
This is naturally the case for the extremal, rank-$1$ error states $\rho_{err}$ from Sec.~\ref{sec_limiting_sc}
for which $\lambda_2 \equiv \lambda_m$ and we are thus guaranteed that the bounds coincide and
are simultaneously saturated.

Third, one can generally expect that the above difference between the largest and smallest error probabilities
is determined by the entropy of the error probability distributions. In particular, ref.~\cite{koczor2020exponential} introduced the
error probability vector $\underline{p}:= (\tfrac{\lambda_2}{1-\lambda}, \tfrac{\lambda_3}{1-\lambda} \dots \tfrac{\lambda_d}{1-\lambda})^T$
and established that the efficacy of the ESD/VD approach depends on the Rényi entropies $H_n(\underline{p})$
of this probability vector. Indeed the difference $\lambda_2 - \lambda_m \leq e^{-H_\infty(\underline{p})}$
generally decays exponentially with the entropy and regardless of the value of $Q$ the difference of the
eigenvalues is negligibly small for high-entropy probability distributions.
One can thus generally expect that for high-entropy experimental states the possible values of the coherent
mismatch are tightly confined by the lower and upper bounds.

\section{Application to quantum circuits \label{circuit_section}}

\subsection{Approximating commutators in noisy quantum circuits}

Let us now consider noisy quantum circuits that prepare quantum states $\rho$
via mappings $\Phi_{c} \rho_{\underline{0}}$ as discussed in Sec.~\ref{sec_prob_def}.
Since the commutator norm $\sigma$ has a special significance (see Sec.~\ref{sec_lower_upper_bounds})
our aim in the following is to approximate the commutator norm for these quantum circuits.

First, let us consider the limiting global worst-case scenario
in which case the ideal unitary computation is followed by a global error channel with probability $\epsilon$ as
 $\Phi_{c} \rho_{\underline{0}} =  (1-\epsilon) \rho_{id} + \epsilon \sum_{j=1}^K M_j \rho_{id} M_j^\dagger$.
This is a special case of Eq.~\ref{error_channel} in which all gates are perfect, except for the last one.
The commutator norm in this case is generally upper bounded as $\sigma^2 \leq \epsilon^2/4 $
and the bound is saturated when the mapping prepares the extremal states in Eq.~\eqref{eq_extremal_state_global}.

Let us now consider the error channel from Eq.~\ref{error_channel} and assume that every
gate has an identical error probability $\epsilon$. Let us now make another simplification for ease of notation
and focus on the case when $K = 1$ for all $k$: such as in case of dephasing noise. 
While these assumptions greatly simplify the following derivations we remark that the present results
can be generalised straightforwardly as discussed in Appendix~\ref{app_general_kraus}.

The considered error model maps the density matrix to an incoherent superposition
(mixture) of $2^{\nu}$ (where $\nu$ is the number of gates) pure states which correspond
to individual error events. For example, the pure state
$U_\nu U_{\nu-1} \cdots  M_k \cdots U_2 U_1  |\underline{0}\rangle$
represents the event where an error happens during the execution of the
$k^{th}$ gate but all other gates are noiseless -- this occurs with probability $\epsilon (1-\epsilon)^{\nu-1}$,
refer to Appendix~\ref{sec_ciruit_commutator} for more details. In general we find that
there are overall $\binom{\nu}{l}$ different events where $l$ errors happen
and each of these have probabilities $\epsilon^l (1-\epsilon)^{\nu-l}$.

As such, we can approximate $\eta$ from
Eq.~\ref{rho_decomp_sum} via the probability that no error happens as
\begin{equation}\label{eq_eta_tilde}
	\tilde{\eta} := (1-\epsilon)^\nu =   (1-\xi /\nu)^\nu \approx e^{-\xi},
\end{equation}
where we have introduced the usual circuit error rate $\xi:=\nu \epsilon$ to denote the expected
number of errors in the full circuit. Indeed, for a sufficiently large number $\nu$ of gates
the probability that no error happens decays exponentially with $\xi$.

We compute the norm (from Statement~\ref{stat_commutator}) of the commutator $[\rho_{id}, \rho]$ in Appendix~\ref{sec_ciruit_commutator}
assuming the above error model and obtain the expression
\begin{equation}\label{summation}
	\sigma^2 	=
	\sum_{\textbf{k} , \textbf{l} \in I}   p_\textbf{k} p_\textbf{l} 	\mathcal{L}_{\textbf{k} \textbf{l}}.
\end{equation}
Here the index set $I$ indexes all distinct error events
and there are exponentially many $|I| = 2^{\nu}-1$ of them. Here,
$p_\textbf{k}$ are probabilities of the individual error events, while $\mathcal{L}_{\textbf{k} \textbf{l}}$
are real numbers that depend on the scalar products between the different erroneous states
and are thus generally upper bounded as $|\mathcal{L}_{\textbf{k} \textbf{l}}| \leq 1$.

The diagonal terms $\mathcal{L}_{\textbf{k} \textbf{k}}$ in the above sum are strictly non-negative and we can
obtain a general upper bound by analytically evaluating the summation  as
\begin{equation} \label{diag_upb}
		f := \sum_{\textbf{k} \in I} p_\textbf{k}^2  =
		(1-\epsilon )^{2 \nu} \left(\left(\frac{1 - 2 (1- \epsilon ) \epsilon }{(1-\epsilon)^2}\right)^\nu-1\right).
\end{equation}

%-------------------------
\begin{figure*}[tb]
	\begin{centering}
		\includegraphics[width=0.85\textwidth]{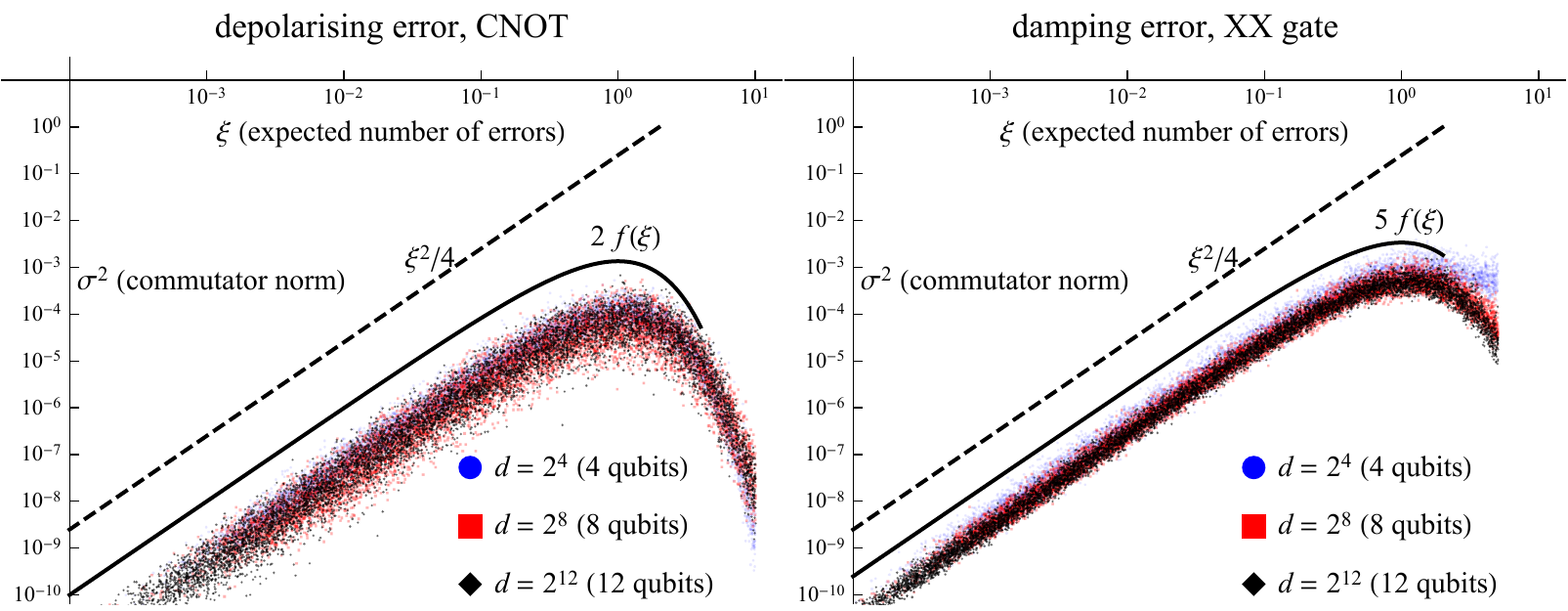}
		\caption{
			Commutator norm $\sigma^2$ in simulated circuits and its upper bound (black solid lines) as a function of the circuit error rate $\xi$.
			Overall $10^4$ circuits composed of $\nu=200$ gates were randomly generated
			as combinations of single qubit $X$ and $Z$ rotations, and CNOT (left) or $XX$ (right) entangling gates. 
			The gates are followed by depolarising (left) or damping (right) noise. 
			It is established in Sec.~\ref{circuit_section} that for sufficiently complex quantum circuits
			the commutator norm $\sigma = \lVert [\rho_{id}, \rho] \rVert_{\infty}$ from Statement~\ref{stat_commutator}
			is upper bounded by the function $f(\xi) \approx  e^{-2 \xi} \, \xi^2 / \nu$ from Eq.~\ref{eq_f_xi_approx} up to a constant,
			where $\xi$ is the expected number of errors in the circuit and $\nu$ is the number of gates.
			For very small $\xi \ll 1$, the bound (black solid lines) is approximately
			by a factor of $\nu$ smaller than the worst-case scenario (dashed black lines).
			$\sigma $ is maximal when $\xi \approx 1/2$ and its maximum
			is at $\sigma_{max} = \mathrm{const} \times \xi/\nu$ which decreases when increasing the number of gates.
			Our approximate upper bound (solid black lines) may break down for very large error rates $\xi \gg 10$.
			Remarkably, in the	practically most important regime $\xi \leq 5$ the same kind of scaling can be observed
			for a large variety of circuits (even for highly deterministic ones) as shown in Fig.~\ref{plot_noisemodel_table}.
			\label{plot_noisemodel}
		}
	\end{centering}
\end{figure*}
%-------------------------

In contrast, the off-diagonal terms in the summation in Eq.~\eqref{summation} depend
on the relative phase between the state vectors of the erroneous quantum states.
We can generally upper bound the summation in Eq.~\eqref{summation} and obtain the completely general
upper bound $\sigma^2  \leq (1- \tilde{\eta})^2$ which is approximated by $\xi^2$ for
small error rates. This bound is indeed pessimistic: even the global worst-case scenario
discussed above has a guaranteed bound $\sigma^2  \leq \xi^2/4$ which is by a 
factor of $4$ smaller.

In order to be able to establish a more meaningful upper bound, we now consider a rather
artificial assumption: we assume that the off-diagonal terms
$\mathcal{L}_{\textbf{k} \textbf{l}}$ with $\textbf{k} \neq \textbf{l}$ in Eq.~\eqref{summation}
are random variables with mean $0$ and some variance $s_{\textbf{k} \textbf{l}}$.
This is equivalent to assuming that complex phases (relative to the ideal state $|\psi_{id} \rangle$)
of the $2^\nu-1$ erroneous pure states uniformly cover the complex plane.
We stress that this assumption is not equivalent to \emph{non-entangling random} circuits
undergoing single-qubit depolarising noise considered in ref.~\cite{huggins2020virtual}.
Those circuits map to noise $\rho_{err} = \mathrm{Id}/d$ that commutes with the ideal state and
indeed one trivially finds that $\sigma=c=0$. In contrast, ref.~\cite{huggins2020virtual} demonstrated that
relatively deep \emph{entangling random} circuits result in a coherent mismatch that is
non-zero and comparable to that of non-random circuits.

The above point can be illustrated via the following analogy:
suppose that we sum up $n$ random real numbers (drawn from a distribution of mean $0$ and variance $s$).
The sum of these numbers is highly unlikely to be $0$. In fact, the result is another random number
that is upper bounded with high probability by some multiple of the square-root of the total
variance that we can compute as $\sqrt{n s}$. In analogy to this observation, we compute the total variance
in Appendix~\ref{sec_ciruit_commutator} and approximately upper bound the summation from  Eq.~\eqref{summation} as
\begin{equation}\label{eq_commut_upb}
	\sigma^2 	=
	\sum_{\textbf{k} , \textbf{l} \in I}   p_\textbf{k} p_\textbf{l} 	\mathcal{L}_{\textbf{k} \textbf{l}}
	\lessapprox
	\mathrm{const} \times f,
\end{equation}
where $f$ was defined in Eq.~\eqref{diag_upb} as the general upper bound on the diagonal entries.
Interestingly, we thus find that assuming randomly distributed off-diagonal entries,
the total sum is only by a constant multiplicative factor larger than the upper bound of
the diagonal entries. Let us now analyse this upper bound.

\subsection{Analysing the approximate bound}

Let us now analyse in detail the upper bound function $f$ from
Eqs.~\eqref{diag_upb}-\eqref{eq_commut_upb}.
In particular, in Appendix~\ref{app_f_analyse} we obtain the approximation
\begin{equation} \label{eq_f_xi_approx}
	f(\xi) \approx  e^{-2 \xi} \, \xi^2 / \nu,
\end{equation}
up to a negligible multiplicative error (that vanishes for large $\nu$) that we neglect for ease of notation.
This approximation is plotted in Fig.~\ref{plot_noisemodel} as a function of
the circuit error rate $\xi$. In the plot one can recognise the following 3 distinct regions.

(a) When the circuit error rate is small $\xi \ll 1$ we find that the upper bound
increases in quadratic order as $\mathrm{const} \times \xi^2 / \nu$. We can compare this expression to the global
worst-case scenario scaling $\xi^2 / 4$ and deduce that the present bound decreases inversely
proportionally with the number $\nu$ of gates (at a fixed error rate $\xi$).
This is illustrated in Fig.~\ref{plot_noisemodel} where the function $f(\xi)$ (solid black lines) are indeed
significantly below the global worst-case bound (dashed black lines), approximately by a factor $\nu$
up to the constant factor from Eq.~\eqref{eq_commut_upb}.

(b)
The maximum of the function $f(\xi)$ is at
\begin{equation}
	\xi_{max} = 1/2 + \mathcal{O}(\epsilon),
\end{equation}
and this position is independent of the constant multiplicative factor from Eq.~\eqref{eq_commut_upb}.
It is also interesting to note
that the global maximum of the function is
\begin{equation*}
	f(\xi_{max}) =  \frac{\epsilon }{2 e} + \mathcal{O}(\epsilon^2)
\end{equation*}
proportional to $\epsilon= \xi/\nu$. This informs us that the maximum of the bound
is decreased inversely proportionally when increasing the number of gates similarly to (a). 
In fact, one can generally state that the upper bound in Eq.~\eqref{eq_commut_upb} scales
as $\sigma^2 = \mathcal{O}(1/\nu)$ for any fixed $\xi$.

(c)
The function  $f(\xi)$ starts to decrease in the third region where  $\xi > 1/2$
and decreases in exponential order asymptomatically for $\xi \gg 1$. 
On the other hand, we observe that in the region where $\xi \gg 1$, our approximation breaks down:
in some instances we numerically observe a different scaling in this regime,
especially when the circuits are highly deterministic.
We have performed additional simulations to illustrate this point: in Fig.~\ref{plot_noisemodel_table}
the commutator norm decreases more slowly for highly deterministic circuits (constant rotation angles
in the quantum gates)
in the region $\xi > 1/2$.   
Nevertheless, this region is not particularly relevant in practice for the following reason.
In the context of the ESD/VD approach the number of circuit repetitions required to suppress
shot noise scales exponentially with the circuit error rate via Eq.~\eqref{eq_eta_tilde}.
It in fact generally holds for error mitigation techniques that their costs grow exponentially
and one thus needs to guarantee a bounded $\xi$. For example, assuming a quadratic (standard shot noise)
scaling of the measurement costs, the overhead at $\xi=5$ is approximately a factor of $2.2\times 10^4$
which is certainly prohibitive in practice \cite{koczor2020exponential,van2020measurement}.

On the other hand, we remarkably find that our bounds hold surprisingly well in all scenarios
in the practically most important region when $\xi \leq 5$. In particular,
these bounds seem to hold remarkably well even for highly deterministic circuits in Fig.~\ref{plot_noisemodel_table},
such as circuits with constant rotation angles -- despite that we assumed randomly distributed
phases for our approximate bounds.
Furthermore, even error models that are beyond the scope of Eq.~\eqref{error_channel}, such as damping in Fig.~\ref{plot_noisemodel},
seem to result in exactly the same kind of scaling.
The numerical data seem to be independent of the number of qubits too
(compare blue, red and black in Fig.~\ref{plot_noisemodel} and in Fig.~\ref{plot_noisemodel_table}) as long as the number
of gates is fixed, which is consistent with the theoretical bounds. 
Most remarkably, up until the point $\xi \approx 1$ each of the large variety
of circuits simulated  in this work resulted in exactly the same type of scaling with respect to $\xi$ and $\nu$
up to only a small (relative to $\nu$) global multiplication factor.

These observations are supported in Appendix~\ref{app_general_kraus} where extensions of our bound to more
general error models are discussed: the form of the upper bound function in Eq.~\eqref{eq_commut_upb} is expected to be 
the same even if one allows higher rank Kraus maps as in Eq.~\eqref{error_channel} or when one allows different
error probabilities for different gates via $\epsilon_k$. Interestingly, if a fraction of the gates commutes with
the error Kraus maps then our bound function $f(\xi)$ still holds up to a minor re-scaling of its argument $\xi$
(via a multiplication with a constant).
Let us now apply our results to bounding the coherent mismatch.

\subsection{Application to coherent mismatch $c$ and noise floor $\sqrt{c}$}
Let us now consider the upper bound for the coherent mismatch via Theorem~\ref{theo_commut_upb}
that depends on the commutator norm. Let us assume that the commutator norm $\sigma$ is bounded
via Eq.~\eqref{eq_commut_upb} and we then obtain 
\begin{equation} \label{scaling_result}
	c \leq  \frac{ \sigma^2  } {  \eta^2 (1- Q)^2 } + \mathcal{O}(\sigma^4)
	\lessapprox \mathrm{const} \times \frac{ \xi^2  } {  \nu (1- Q)^2 } ,
\end{equation}
where we have used that the probability of the ideal state $\eta$ is upper bounded as $\eta \geq \tilde{\eta}$ via
Eq.~\eqref{eq_eta_tilde}. Let us now remark on 3 important consequences of the above approximate bound
and how it confirms prior numerical observations.

(a) Eq.~\eqref{scaling_result} establishes that the coherent mismatch scales as $c = \mathcal{O}(\nu \epsilon^2)$
when assuming a fixed $Q$, where $Q$ was defined in Theorem~\ref{theo_commut_upb} as the ratio of the two
largest eigenvalues.
This scaling is consistent with previous numerical observations:
It was numerically observed in ref.~\cite{koczor2020exponential} (ref.~\cite{huggins2020virtual}) that
if one increases the per-gate error probability $\epsilon$
in a fixed quantum circuit then the coherent mismatch (noise floor) grows quadratically (linearly)
as $c = \mathcal{O}(\epsilon^2)$ ($\sqrt{c} = \mathcal{O}(\epsilon)$).

(b) Eq.~\eqref{scaling_result} establishes a scaling $c = \mathcal{O}(\nu)$ when increasing the number $\nu$
of gates at a fixed per-gate error rate.
This is consistent with the observation of ref.~\cite{koczor2020exponential} that increasing the
number of gates in a circuit of fixed per-gate error probability $\epsilon$ increases the coherent mismatch
proportionally as $c = \mathcal{O}(\nu)$, while ref.~\cite{huggins2020virtual}  similarly
observed in numerical random-circuit simulations that the noise floor $\sqrt{c}$
slightly increases when increasing $\nu$.
As noted in the above section, the scaling results in this work were derived assuming sufficiently
complex quantum circuits, but these results appear to hold remarkably well for even
highly deterministic circuits too as long as the circuit error rate does
not significantly exceed $\xi \approx 5$.

The crucial implication of this scaling for practical applications is the following.
Consider a computational task that is defined for $N$ qubits. A quantum circuit of depth $a(N)$ then requires
overall $\nu = \mathcal{O}(N a(N))$ gates to implement the computation. This ensures
us that the coherent mismatch decreases even for constant depth as $c = \mathcal{O}(\xi^2 N^{-1})$
when the size of the computation (via $N$) is increased at a constant circuit error rate $\xi$.
In practice one needs to keep $\xi$ at least bounded to ensure a bounded sampling cost which was
discussed in the previous section.

(c) Another important consequence of these scaling results is the following. 
Recall that the probability that no error happens decays exponentially with the
circuit error rate as $\tilde{\eta} \approx e^{-\xi}$. This is approximately constant for
a fixed value of $\xi$. In stark contrast, we have found that the coherent mismatch depends
on the number of gates and scales as $c=\mathcal{O}(\xi^2/\nu)$.
Let us now compare the fidelity $F$ that decreases due to incoherent errors and the coherent
fidelity $1-c$ that decreases due to the coherent mismatch in the dominant eigenvector.
The fidelity can be approximated as $F \approx \tilde{\eta} \approx  e^{-\xi} $ and decays exponentially
due to incoherent errors, while the fidelity $1-c$ due to the coherent mismatch decays as
$1-c = 1-\mathcal{O}(\xi^2/\nu)$.
The ratio of these two fidelities can then be approximated as
\begin{equation*}
	\frac{1-c}{F} \approx e^{\xi} - \mathcal{O}(e^{\xi} \xi^2/\nu).
\end{equation*}
Indeed the above ratio increases exponentially when increasing $\xi$ within a finite range, e.g., when
$\xi < 10$ and when the number $\nu$ of gates is sufficiently large. This is consistent
with numerical observations of ref.~\cite{koczor2020exponential}:
increasing the number of gates in a sufficiently complex circuit decreases the
incoherent fidelity ($F$) exponentially faster than it decreases the coherent fidelity ($1-c$).
Very importantly, this ensures us that the coherent mismatch of the dominant
eigenvector (which cannot be suppressed) causes an exponentially smaller error when compared to
the incoherent decay of the fidelity $F$. Here the latter can indeed be suppressed exponentially
by increasing the number of copies in the ESD/VD approach.

\section{Discussion and Conclusion \label{discussion}}

The present work considered the fundamental question: given a noisy quantum state,
how well does its dominant eigenvector $| \psi\rangle$ approximate a corresponding
ideal, noise-free computation $| \psi_{id} \rangle$? While it is of fundamental importance
to understand how noise affects quantum systems, this particular question has crucial
practical relevance. The recently introduced ESD/VD error suppression techniques
are ultimately limited by the coherent mismatch.

This work has established general upper bounds and scaling results for the coherent mismatch 
and presented a comprehensive analysis of its implications in practically relevant scenarios.
As such, it was established that the coherent mismatch is indeed negligibly small
for sufficiently complex noisy quantum circuits, typically used in variational quantum algorithms
and other near-term quantum algorithms
\cite{endo2020hybrid,cerezo2020variationalreview,bharti2021noisy}. It is interesting to note that
since variational quantum algorithms rely on optimising a cost function, this optimisation
can be expected to anyway minimise the effect of the coherent mismatch.
Let us briefly summarise the most important results.

(a) The bound based on the noise floor $\sqrt{c}$ in ref.~\cite{huggins2020virtual} was improved
and quadratically smaller bounds are obtained for the pivotal case of preparing eigenstates
-- see Sec.~\ref{precision_motivation}.

(b) A general upper bound for the coherent mismatch was obtained in Sec.~\ref{sec_upbdelta} by explicitly constructing
worst-case scenario extremal quantum states that saturate it
(for this we analytically computed the coherent mismatch in Sec.~\ref{sec_arrowhead_eigenvector}
using our arrowhead decomposition obtained in Sec.~\ref{sec_arrowhead}). The present problem is closely
related to an important problem in mathematics: bounding the eigenvalues of a sum of two matrices (Weyl inequalities).
While those bounds are well-known to be saturated by identical dominant eigenvectors, it was shown
in Sec.~\ref{sec_limiting_sc} that bounds obtained in this work are in stark contrast saturated
by the close-to-orthogonal dominant eigenvectors of the extremal quantum states.

(c) In the ESD/VD approach, even for extremal quantum states, one needs at least $3-4$ copies of
the noisy state to suppress errors to the noise floor set by the coherent mismatch, 
see Sec.~\ref{numb_copies}.
The coherent mismatch is thus guaranteed to be negligible in practical applications where
the quantum device is limited in its ability to prepare a large number of copies.

(d) Another closely related problem in mathematics is upper bounding the matrix norm of the commutator
of two matrices. We obtained considerably tighter bounds then prior results in the specific case of
the matrix norm of the commutator between two density matrices, see Sec.~\ref{sec_comm_norm}.
Interestingly, the commutator norm is given by the generalised quantum-mechanical variance of the density matrix
which quantity is also proportional to the quantum Fisher information.

(e) General upper and lower bounds were obtained in Sec.~\ref{sec_commut_upb} and Sec.~\ref{sec_lower_bound}
for the coherent mismatch in terms of the commutator norm from (d). It was established that in the practically important region the
upper and lower bounds are close to each other and thus tightly confine possible values of $c$ 
-- while the bounds asymptomatically coincide. It was also shown that the coherent mismatch generally
decays exponentially with Rényi entropies of the error probabilities -- indeed, similar scaling results
were obtained in ref.~\cite{koczor2020exponential} for the efficacy of the ESD/VD approach and it was
noted that near-term quantum devices are expected to produce high-entropy errors.

(f) We finally applied the above general results to the specific but pivotal case of noisy quantum circuits
in Sec.~\ref{circuit_section}.
The resulting approximate bounds confirmed scaling results of ref.~\cite{koczor2020exponential}:
the coherent mismatch in sufficiently complex noisy circuits is decreased inversely proportionally
when increasing the size of the computation (by increasing the number of qubits at a fixed error rate).
Furthermore, in the practically important regions, the incoherent deterioration of a quantum state
is exponentially more severe than the drift in the dominant eigenvector. This establishes
that the coherent mismatch is indeed negligible in relevant applications of the ESD/VD approach.

Results obtained in this work pave the way towards developing advanced error mitigation techniques that
will be crucial for the successful exploitation of noisy quantum devices. A number of apparent
questions will be worth investigating in the future, such as developing twirling techniques
(and generalisations thereof) that
potentially decrease the coherent mismatch without affecting the ideal part of
the computation. In particular, one could obtain a series of quantum circuits $\Phi_c^{(l)}$
whose unitary component $U_c$ is identical for every $l$ while the noise component is
different. The average of such channels $|L|^{-1} \sum_{l\in L}\Phi_c^{(l)}$ is thus guaranteed
to increase the entropy of errors resulting in a smaller coherent mismatch.

Another open question is related to similar themes in mathematics: 
Analogously to the Weyl inequalities for the eigenvalues, 
is it possible to generalise the present results to obtain a series of upper and lower bounds
for infidelities in all eigenvectors (not just the dominant one)? Answering this question will be highly non-trivial since
the generalisation to arbitrary matrices will require to go beyond the analytical expressions obtained
for $c$ and $\sigma$ which assumed that $\rho_{id}$ is rank-1 thus having only a single dominant component.

Let us finally remark that arguments presented in this work naturally generalise to infinite-dimensional
quantum states $\rho$ as general trace-class operators.

\section*{Acknowledgments}
I would like to thank Simon C. Benjamin, Earl Campbell and Sam McArdle for useful discussions.
I would like to thank Robert Zeier, Zhenyu Cai and Adrian Chapman for their valuable comments
and for carefully reading drafts of this work.
I acknowledge funding received from EU H2020-FETFLAG-03-2018 under the grant
agreement No 820495 (AQTION) and from EPSRC Hub grant under the agreement number
EP/T001062/1.
I acknowledge financial support from the Glasstone Research Fellowship of the University of Oxford.
The numerical modelling involved in this study made
use of the Quantum Exact Simulation Toolkit (QuEST), and the recent development
QuESTlink\,\cite{QuESTlink} which permits the user to use Mathematica as the
integrated front end. I am grateful to those who have contributed
to both these valuable tools.

%\bibliography{bibliography}


\onecolumngrid
\appendix

\section{Validity of the decomposition in Eq.~\eqref{rho_decomp_sum} \label{def_state_appendix}}
Here we discuss the scope and non-uniqueness of the decomposition in Eq.~\eqref{rho_decomp_sum}.
Let us remark that this decomposition is very useful for illustrating and understanding the core problem
while it is also natural in most of the typical error channels.

Let us first note that the `quality' of the noisy quantum state is expressed via the fidelity
$F:= \langle \psi_{id} | \rho | \psi_{id} \rangle$, which can be interpreted as
a probability; indeed we need to restrict the mapping $\Phi_{c}$ to ones that
result in $F>0$ in order to exclude trivial cases. In case if $\rho$ is full-rank,
then there always exists a decomposition $\rho =  \eta \rho_{id} + (1-\eta)  \rho_{err}$
for some $\eta > 0$ and for positive semi-definite $\rho_{err}$. This can be shown straightforwardly
by subtracting  $\rho - \eta \rho_{id}$ since the difference matrix is generally guaranteed (due to the Weyl inequalities
in Sec.~\ref{sec_related_maths}) to be positive semi-definite as
long as $\eta \leq \lambda_m$, where $\lambda_m$ is the smallest eigenvalue of $\rho$.

While the considered decomposition is natural in case of many of the typical error channels, e.g.,
the one considered in Eq.~\eqref{error_channel}, it is not unique and multiple values of $\eta$ can satisfy it.
Nevertheless, we can uniquely define an optimal $\eta$ via the following optimisation problem as
\begin{equation}\label{optimal_eta}
	\eta = \max \{ \eta \,  | \, \frac{\rho - \eta \rho_{id}}{1-\eta}\,  \text{is positive-semidefinite}\}.
\end{equation}
In the above equation, we find the largest possible $\eta$ for which the resulting operator
still corresponds to a valid density matrix.
This definition would guarantee that the parameter $\delta$ in Definition~\ref{def_state} is minimal under the above
decomposition and the resulting upper bounds in Theorem~\ref{theo_upbdelta} are the least possible.

In summary, the decomposition in Eq.~\eqref{rho_decomp_sum} is guaranteed to exist for full-rank density matrices $\rho$,
but does not necessarily exist for arbitrary density matrices. An example when the decomposition does not exist is
when $\rho = | \chi \rangle \langle \chi |$ and $| \chi \rangle \neq | \psi_{id}\rangle$, which is the case of a purely
coherent error.
Another disadvantage is that the decomposition in Eq.~\eqref{rho_decomp_sum} is not unique since multiple
values of $\eta$ can satisfy it: we have defined an optimal value of $\eta$ above which, however, requires
a non-trivial optimisation. Nevertheless, the arguments presented in, e.g., Sec.~\ref{sec_commut_upb}, which depend on the commutator
norm are completely independent of this decomposition and apply to any density matrix (even to rank-deficient ones).

\section{Noise floor and coherent mismatch \label{app_noise_floor}}

%-------------------------
\begin{figure*}[tb]
	\begin{centering}
		\includegraphics[width=0.95\textwidth]{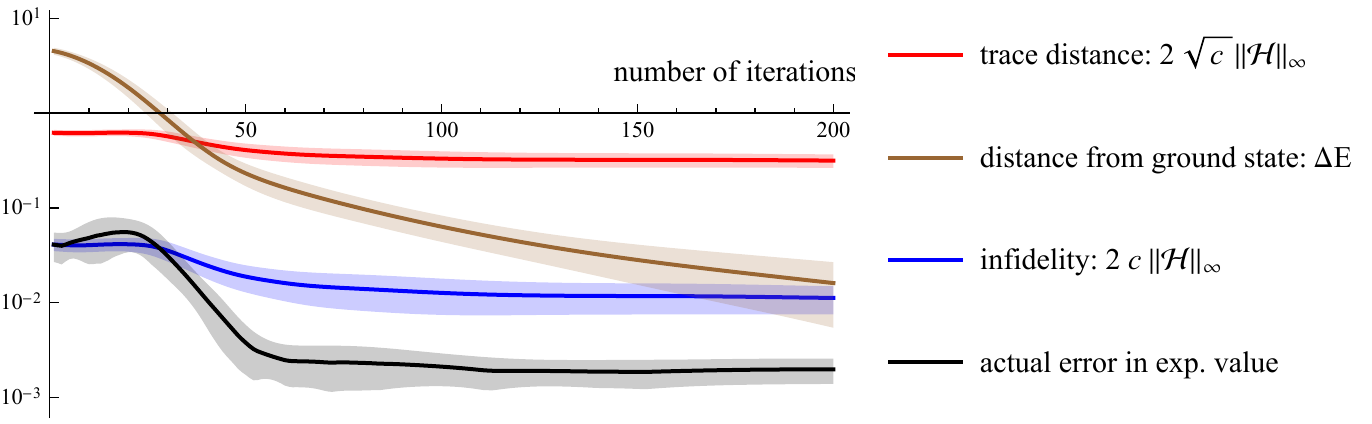}
		\caption{
			Example of a variational quantum optimisation using $8$ qubits. The ground state of a spin-ring Hamiltonian
			with nearest neighbour $XX$, $YY$ and $ZZ$ couplings and randomly generated on-site
			frequencies $\omega_k Z$ is searched via a VQE optimisation.
			The distance from the exact ground-state energy (brown) approaches $0$ as the number of iterations
			is increased. If the errors in the noisy quantum circuit (circuit error rate $\xi \approx 2$)
			are suppressed via the ESD/VD approach then one can measure the exact expectation value with respect
			to the dominant eigenvector of the noisy quantum state -- this causes an error (black) when
			compared to the ideal expectation value in a noiseless circuit. This error is
			generally upper bounded by the noise floor (red, trace distance) which is very pessimistic and as the
			quantum state approaches the ground state then the error is guaranteed to be upper bounded
			by the coherent mismatch (blue, infidelity). The latter bound seems to hold even for approximate ground states
			(low iteration depth).
			\label{plot_optimise}
		}
	\end{centering}
\end{figure*}
%-------------------------

\begin{proof}
	Let us consider the expression for the noise floor
\begin{equation*}
	N_f := \lim_{n \rightarrow \infty} T(\rho_n, \rho_{id}).
\end{equation*}
Using notations in Eq.~\eqref{rho_decomp} we can express the state as
\begin{equation*}
	 \rho_n:= \frac{\rho^n}{\tr[\rho^n]}
	 = \frac{\sum_{k=1}^d \lambda_k^n | \psi_k \rangle \langle \psi_k |}{\sum_{k=1}^d \lambda_k^n } 
	 =   \frac{ | \psi \rangle \langle \psi |+  \sum_{k=2}^d \lambda_k^n/\lambda^n | \psi_k \rangle \langle \psi_k |}{ 1+ \sum_{2=1}^d \lambda_k^n/\lambda^n}
	 = \xi(n) | \psi \rangle \langle \psi | + M_n,
\end{equation*}
where  $\xi(n):= [1+ \sum_{2=1}^d \lambda_k^n/\lambda^n]^{-1}$ which exponentially converges to its limit $\lim_{n \rightarrow \infty} \xi(n) = 1$
and the residual matrix  $ M_n:=  \xi(n)  \sum_{k=2}^d \lambda_k^n/\lambda^n | \psi_k \rangle \langle \psi_k |$
is diagonal with eigenvalues $ \xi(n) \lambda_k^n/\lambda^n$. This ensures us that in any $p$-norm topology the matrix
$M_n$ converges in exponential order to its limit $ \lim_{n \rightarrow \infty} \lVert M_n \rVert_p = 0$.
We can thus deduce that in any matrix norm topology the distilled matrix $\rho_n$ approaches the pure state $| \psi \rangle \langle \psi |$
in exponential order. We thus find that 
\begin{equation*}
	 \lim_{n \rightarrow \infty} T(\rho_n, \rho_{id}) = T(| \psi \rangle \langle \psi |, \rho_{id}) = \sqrt{1- |\langle \psi_{id} | \psi \rangle|^2  }
	 = \sqrt{c},
\end{equation*}
where in the second equality we have used that the trace distance of two pure states can be evaluated analytically in terms of the fidelity.

One can straightforwardly show that the trace distance upper bounds measurement 
errors with respect to any bounded observable $O$ as
\begin{equation*}
 |\tr[O \rho_{id}] -	\tr[O \rho]| = | \tr[O  (\rho_{id} - \rho) ]| = | \sum_k d_k  \langle \chi_k | O | \chi_k \rangle |
 \leq \lVert O \rVert_\infty \sum_k |d_k|  = 2 \lVert O \rVert_\infty \,  T(\rho, \rho_{id})
\end{equation*}
where $ d_k $ and $ | \chi_k \rangle $ are eigenvalues and eigenvectors of the difference of
the two density matrices.

The quantity $2\sqrt{c} \lVert O \rVert_\infty$ thus upper bounds the measurement error of any
bounded observable. Interestingly, if the ideal computational state approximates an eigenvector
of the measurement operator we then find the following. Let us write the dominant eigenvector
as a linear combination of two vectors
\begin{equation*}
	| \psi \rangle = \sqrt{1-c} | \psi_{id}\rangle   + \sqrt{c} | \psi_{\perp}\rangle,
\end{equation*}
with real, non-negative $c$ (since we are free to choose the global phase of a state vector).
It follows that the measurement of an observable yields
\begin{equation*}
	\langle \psi | O | \psi \rangle = (1-c) \langle \psi_{id} | O | \psi_{id} \rangle 
	+ 2 \sqrt{c(1-c) } \mathrm{Re} \langle \psi_{\perp} | O | \psi_{id} \rangle 
	+ c \langle \psi_{\perp} | O | \psi_{\perp} \rangle.
\end{equation*}
In the special case when $O | \psi_{id} \rangle  = E | \psi_{id} \rangle $ for some 
real $E$ then we obtain $ \langle \psi_{\perp} | O | \psi_{id} \rangle = 0$
and finally, the measurement error of the observable is
\begin{equation*}
	| 	\langle \psi | O | \psi \rangle  -	\langle \psi_{id} | O | \psi_{id} \rangle | = 
	c ( \langle \psi_{\perp} | O | \psi_{\perp} \rangle - \langle \psi_{id} | O | \psi_{id} \rangle )
	\leq 2 c  \lVert O \rVert_\infty.
\end{equation*}

\end{proof}

\section{Proof of Statement~\ref{stat_arrowhead} \label{proof_arrowhead}}
\begin{proof}
	We compute the matrix representation of the operator $\tilde{\rho}$ by choosing an orthonormal basis
	that defines the unitary transformation $U$ such that $U \rho U^\dagger = \tilde{\rho}$.
	Let us choose the leading basis vector as $|\psi_{id}\rangle$ and thus $U |\psi_{id}\rangle =: | \tilde{\psi}_{id}\rangle = (1,0, \dots 0)^T$.
	We can choose the rest of the basis vectors $ |\phi_k\rangle$ arbitrarily as long as
	$\langle \psi_{id} |\phi_k\rangle =0$ for all $k=\{2,3,\dots d\}$.  We define $ |\phi_k\rangle$ such that
	they are eigenvectors of $P \rho P$, where $P = \mathrm{Id}-|\psi_{id}\rangle \langle \psi_{id}|$
	projects onto the orthonormal subspace.
	Furthermore, we are free to choose the global phase of the basis vectors and we note that this global phase has no effect on the diagonal
	entries since
	\begin{equation*}
		D_k :=\langle e^{i\theta_k} \phi_k | \rho | e^{i\theta_k} \phi_k \rangle = \langle \phi_k | \rho |  \phi_k \rangle \geq 0.
	\end{equation*}
	Here $D_k$ are non-negative since  $\rho$ is by definition positive semi-definite.
	We can implicitly define the global phase of the vectors $|\phi_k \rangle$ such that the off-diagonal entries
	are real and non-negative as
	\begin{equation*}
	C_k :=\langle  \psi_{id} | \rho |  \phi_k \rangle \in \mathbb{R} \quad \text{since} 
	\quad C_k e^{i\theta_k}  =\langle  \psi_{id} | \rho | e^{i\theta_k} \phi_k \rangle,
	\end{equation*}
	We have thus established a matrix representation of $\tilde{\rho}$ such that $D_k, C_k \in \mathbb{R}$ and $D_k, C_k \geq 0$,
	and $\tilde{\rho}$ is diagonal in the subspace orthogonal to $|\psi_{id}\rangle$.
	We can finally explicitly write the arrowhead matrix using the above established
	orthonormal basis $\{ \psi_{id}, \phi_2, \phi_3 \dots \phi_n \}$ that defines the unitary
	transformation $U$ such that 
	\begin{equation*}
	U \rho U^\dagger = \tilde{\rho}  = \begin{pmatrix}
			F & C_2 & C_3 & \dots & C_{d}\\
			C_2 & D_2 &  &  &  \\
			C_3 &  & D_3 \\
			\vdots & & & \ddots &\\
			C_{d} &  &  & \dots & D_{d}
		\end{pmatrix}.
	\end{equation*}
\end{proof}

\section{Proof of Statement~\ref{stat_coherent} \label{proof_coherent}}

\begin{proof}
	If we explicitly know the arrowhead matrix, then its eigenvectors can be computed analytically~\cite{o1990computing},
	refer also to Eq.~(5) in \cite{gu1995divide}.
	Recall that we introduced  the orthonormal basis $\{ \tilde{\psi}_{id}, \phi_2, \phi_3 \dots \phi_n \}$
	in Appendix~\ref{proof_arrowhead} and used it to represent $\rho$ as an arrowhead matrix.
	This corresponds to a unitary transformation $U \rho U^\dagger = \tilde{\rho}$, where
	$\tilde{\rho}$ is the arrowhead matrix from 	Statement~\ref{stat_arrowhead}.
	Using these notations we can write the dominant eigenvector of $\rho$
	from  Definition~\ref{def_state} up to this unitary transformation as
	\begin{equation*}
	U | \psi \rangle  = | \tilde{\psi} \rangle  = [1 + \sum_{k=2}^d \frac{C_k^2}{(D_k - \lambda)^2}]^{-1/2} 
	\,
	(-1, \frac{C_2}{D_2 - \lambda}, \frac{C_3}{D_3 - \lambda} \dots \frac{C_d}{D_d - \lambda}) ^T,
	\end{equation*}
	where $\lambda$ is the domiant eigenvector from Definition~\ref{def_state}.
	
	We can apply this explicit formula to compute the coherent mismatch from Definition~\ref{def_state} as
	\begin{equation}
		c := 1- |\langle \psi_{id} | \psi \rangle|^2 = 1- |\langle \tilde{\psi}_{id} | \tilde{\psi} \rangle|^2
		 = 1- [1 + \sum_{k=2}^d \frac{C_k^2 }{(D_k - \lambda)^2 } ]^{-1},
	\end{equation}
	where we have used that $| \tilde{\psi}_{id}\rangle = (1,0, \dots 0)^T$.
	\end{proof}

\section{Proof of Remark~\ref{remark_perturb} \label{proof_perturb}}
\begin{proof}
	The first-order perturbation correction to the dominant eigenvector can be computed via Statement~\ref{stat_arrowhead} 
	using the arrowhead decomposition as
	\begin{equation*}
		\tilde{\rho} = F |\tilde{\psi}_{id}\rangle\langle \tilde{\psi}_{id}| + D +C,
	\end{equation*}
	where $|\tilde{\psi}_{id}\rangle = (1,0, \dots 0)^T$ and $D$ is diagonal. Let us now treat $C$
	as a perturbation of the diagonal matrix $F |\tilde{\psi}_{id}\rangle\langle \tilde{\psi}_{id}| + D$
	and use the usual perturbative expansion, see e.g.,  Eq.~(5.1.44) in \cite{sakurai1995modern}.
	We can thus compute the first-order correction to the dominant eigenvector
	(and recall that $| \tilde{\psi}\rangle := U | \psi\rangle$) as
	\begin{equation}
		| \tilde{\psi}_{corr}^{(1)} \rangle = (0, \frac{C_2}{F-D_2} , \frac{C_3}{F-D_3} ,\dots \frac{C_d}{F-D_d} ).
	\end{equation}
	The normalised first order eigenvector is obtained as
	\begin{equation*}
		| \tilde{\psi}^{(1)} \rangle = [1+ \sum_{k=2}^d \frac{C_k}{F-D_k} ]^{-1/2}
		\,
		(1, \frac{C_2}{F-D_2} , \frac{C_3}{F-D_3} ,\dots \frac{C_d}{F-D_d} ).
	\end{equation*}
	Computing the coherent mismatch $c$ from the above first-order perturbation we obtain
	\begin{equation*}
		c_{pert} :=  1- |\langle \psi_{id} | \psi^{(1)} \rangle|^2 = 1- |\langle \tilde{\psi}_{id} | \tilde{\psi}^{(1)} \rangle|^2
		= 1 - [1 + \sum_{k=2} \frac{C_k^2}{(F-D_k)^2}]^{-1}.
	\end{equation*}

Let us remark that using Eq~(10.2) from \cite{wilkinson1965algebraic}, one can obtain the more accurate first-order
approximation assuming explicit knowledge of the eigenvalues 
\begin{equation*}
	c_{pert} = 1 - [1 + \sum_{k=2}^{d} \frac{C_k^2}{(\lambda-\lambda_k)^2}]^{-1}
\end{equation*}
\end{proof}

\section{Proof of Theorem~\ref{theo_upbdelta} \label{proof_upbdelta}}
\begin{lemma}
	The coherent mismatch $c$ of density matrices is generally bounded
	via their arrowhead-matrix representations with non-negative entries $F, C_k, D_k \geq 0$ with
	$2 \leq k\leq d$ from Statement~\ref{stat_arrowhead} as
	\begin{equation}
		\label{general_upper_bound}
	1- [1 + \frac{  \lVert \mathcal{C} \rVert^2  }{(\lambda - D_m)^2 } ]^{-1}	\leq c  \leq  1- [1 + \frac{  \lVert \mathcal{C} \rVert^2  }{(\lambda - D_2)^2 } ]^{-1},
	\end{equation}
where $ \lVert \mathcal{C} \rVert^2 := \sum_{k=2}^d C_k^2 $ and  $D_m$ is the smallest non-zero diagonal entry of the arrowhead matrix.
The upper bound is saturated by any density matrix $\rho$ that can be mapped to an arrowhead matrix of the form
\begin{equation}\label{extremal_state_matrix}
	A_{max} := M_{max} \oplus \mathrm{diag}( D_3, D_4, \dots D_d ) =
		\begin{pmatrix}
		F & C_2 & 0 & \dots & 0 \\
		C_2 & D_2  &  &  &  \\
		0 &  & D_3 \\
		\vdots & & & \ddots &\vdots \\
		0 &  &  & \dots & D_d
	\end{pmatrix},
\end{equation}
where the only non-zero off-diagonal entry $C_2$ is next to $D_2$.
Furthermore, $D_k$ with $k>2$ are eigenvalues of the arrowhead matrix.
The lower bound is saturated by analogous matrices but the non-trivial $2$-dimensional subspace
$M_{min} = 		\begin{pmatrix}
	F & C_m  \\
	C_m & D_m \\
\end{pmatrix}$
contains the smallest non-zero diagonal entry $D_m>0$ as
\begin{equation}\label{extremal_state_matrix_lower}
	A_{min} := M_{min} \oplus  \mathrm{diag}(  \{D_2, D_3,  \dots D_d\} \setminus \{D_m\} ) =
	\begin{pmatrix}
		F & C_m & 0 & \dots & 0 \\
		C_m & D_m  &  &  &  \\
		0 &  & D_2 \\
		\vdots & & & \ddots &\vdots \\
		0 &  &  & \dots & D_{d}
	\end{pmatrix}.
\end{equation}
\end{lemma}

\begin{proof}

\noindent  \textbf{Upper bound\\}	
Let us consider arrowhead matrices with arbitrary non-negative entries $F, C_k, D_k \geq 0$ with
$2 \leq k\leq d$, which contain the density matrices from Statement~\ref{stat_arrowhead}.
Recall from Statement~\ref{stat_coherent} that the coherent mismatch can be expressed as
		\begin{equation}
			c(\Xi) = 1- |\langle \psi_{id} | \psi \rangle|^2 = 1- [1 + \Xi ]^{-1},
		\end{equation}	
		where we have used the notation  $	\Xi:=	\sum_{k=2}^d \frac{C_k^2 }{(\lambda - D_k)^2 } $.
We can upper bound $\Xi$ by using the interlacing property with $\lambda \geq D_2 \geq D_3 \dots \geq D_d$ as
\begin{equation}
	\Xi \leq 	\sum_{k=2}^d \frac{C_k^2 }{(\lambda - D_2)^2 } =  \frac{  \lVert \mathcal{C} \rVert^2  }{(\lambda - D_2)^2 },
\end{equation}
where we have introduced the $d-1$-dimensional vector $\mathcal{C} := (C_2,C_3, \dots C_N)^T$.

The upper bound is saturated by arrowhead matrices of the form	
\begin{equation}
	A_{max} = \begin{pmatrix}
		F & \mathcal{C}^T &  & \dots & \\
		\mathcal{C} & D_2 \,  \mathrm{Id}_{\nu} &  &  &  \\
		 &  & 	D_{\nu+2} \\
		\vdots & & & \ddots &\vdots \\
		&  &  & \dots & D_{d}
	\end{pmatrix},
\end{equation}
where we used the notation $\mathcal{C}:= (1,1,1, \dots 1)\lVert \mathcal{C} \rVert/\sqrt{\nu}$ and $\nu$ is the
dimension of the identity matrix $\mathrm{Id}_{\nu}$.  It is straightforward to show that these matrices saturate the
upper bound just by computing the coherent mismatch as $\Xi = 	\sum_{k=2}^d \frac{C_k^2 }{(\lambda - D_2)^2 } $
which coincides with the upper bound above.
The above matrix has non-zero off-diagonal entries in the upper left corner in a $\nu+1$-dimensional subspace.
Here $\nu$ represents the degeneracy of the eigenvalues of the matrix $A_{max}$, and the upper bound
in Eq.~\eqref{general_upper_bound} is saturated by any such matrix with any $\nu \leq d-1$.
For example, setting $\nu=1$ assumes no degeneracy of the eigenvalues.
Let us now express this non-trivial subspace explicitly as
\begin{equation}
\begin{pmatrix}
		F & \mathcal{C}^T  \\
	\mathcal{C} & D_2 \,  \mathrm{Id}_{\nu}   \\
\end{pmatrix}
=
 \begin{pmatrix}
	F & q & q & \dots & q \\
	q & D_2  &  &  &  \\
	q &  & D_2 \\
	\vdots & & & \ddots &\vdots \\
	q&  &  & \dots & D_2
\end{pmatrix},
\end{equation}
where $q:=\lVert \mathcal{C} \rVert/\sqrt{\nu}$.
The last step is that we show that the above matrix
is unitarily equivalent to the matrix
\begin{equation}\label{unit_equiv_matrices}
	\begin{pmatrix}
		F & q & q & \dots & q \\
		q & D_2  &  &  &  \\
		q &  & D_2 \\
		\vdots & & & \ddots &\vdots \\
		q&  &  & \dots & D_2
	\end{pmatrix}
\simeq
	\begin{pmatrix}
	F & \sqrt{\nu}\,  q & 0 & \dots & 0 \\
	\sqrt{\nu} \,  q & D_2  &  &  &  \\
	0 &  & D_2 \\
	\vdots & & & \ddots &\vdots \\
	0 &  &  & \dots & D_2
\end{pmatrix},
\end{equation}
where only a $2$-dimensional sub-block has non-zero off-diagonal entries
and $\simeq$ represents that the two matrices are unitarily equivalent.
Note that we map density matrices to arrowhead matrices by applying a suitable 
unitary transformation. Let us consider the following example. Instead of mapping $\rho$ to
the matrix on the left-hand side above by applying $U_1$, we map $\rho$ to the matrix
on the right-hand side by applying $U_2 U_1$, where $U_2$ maps between the above two matrices
due to their unitary equivalence. We will prefer to map density matrices to the arrowhead matrices
on the right-hand side, albeit the two forms would be equivalent and result in the same coherent mismatch.

The most straightforward way to show the unitary equivalence of the above two matrices is
by recognising that both are arrowhead matrices and their eigenvalues are
roots of the same secular function from Eq.~\eqref{eq_eigenvalue}
\begin{equation} 
	P(x) = x - F +\frac{ \nu\,  q^2 }{(D_2 - x) }.
\end{equation}
As the two matrices share the same eigenvalues, there exists a
unitary transformation that transforms one into the other.
It follows therefore that the upper bound on the coherent mismatch from
Eq.~\eqref{general_upper_bound} is saturated by any density matrix that can be mapped to an
arrowhead matrix of the form
\begin{equation}
	A_{max} =	\begin{pmatrix}
		F & C_2 & 0 & \dots & 0 \\
		C_2 & D_2  &  &  &  \\
		0 &  & D_3 \\
		\vdots & & & \ddots &\vdots \\
		0 &  &  & \dots & D_d
	\end{pmatrix},
\end{equation}
with the diagonal entries satisfying the usual ordering $D_2 \geq D_3 \geq \dots D_d$
and it follows that $D_k$ with $k>2$ are eigenvalues of the arrowhead matrix.
The degenerate case is recovered via $D_2 = D_3 = \dots D_{\nu+1}$.\\

\noindent \textbf{Lower bound\\}
We consider arrowhead matrices with arbitrary non-negative entries $F, C_k, D_k \geq 0$ with
$2 \leq k\leq d$, which contain the density matrices from Statement~\ref{stat_arrowhead}.
The lower bound on $\Xi$ can be obtained as
\begin{equation*}
	  \frac{  \lVert \mathcal{C} \rVert^2  }{(\lambda - D_d)^2 } \leq \Xi,
\end{equation*}
where $D_d$ is the smallest diagonal entry. Arrowhead matrices that saturate the lower
bound can be constructed by following a very similar argument to the one presented above.
We find that any density matrix that can be mapped to an arrowhead matrix of the following form
saturates the lower bound
\begin{equation}
	\begin{pmatrix}
		F & C_d & 0 & \dots & 0 \\
		C_d & D_d  &  &  &  \\
		0 &  & D_2 \\
		\vdots & & & \ddots &\vdots \\
		0 &  &  & \dots & D_{d-1}
	\end{pmatrix}.
\end{equation}
One additional remark is that the non-trivial 2-dimensional sub-block 
needs to be positive-semidefinite as the matrix represents a density matrix. 
It therefore follows that only arrowhead matrices with $C_d \leq \sqrt{D_d F}$ 
can represent valid density matrices. We can exclude trivial cases
such as $D_d = 0$, where necessarily $C_d = 0$, and tighten the lower bound the following way. Density matrices
that are mapped to arrowhead matrices satisfy a lower bound on the coherent mismatch via
\begin{equation*}
	\frac{  \lVert \mathcal{C} \rVert^2  }{(\lambda - D_m)^2 } \leq \Xi,
\end{equation*}
where $D_m$ is the smallest non-zero eigenvalue. This lower bound is saturated by
density matrices that can be mapped to arrowhead matrices of the form
\begin{equation} 
	A_{min} =
	M_{min} \oplus  \mathrm{diag}(  \{ D_2, D_4 \dots D_d \} \setminus \{D_m\}) =
		\begin{pmatrix}
		F & C_m & 0 & \dots & 0 \\
		C_m & D_m  &  &  &  \\
		0 &  & D_2 \\
		\vdots & & & \ddots &\vdots \\
		0 &  &  & \dots & D_{d}
	\end{pmatrix}.
\end{equation}\\
\end{proof}

Let us now prove Theorem~\ref{theo_upbdelta}.

\begin{proof}	
	
\noindent \textbf{Explicit construction of extremal density matrices}\\
It was shown above that the upper bound of the coherent mismatch is saturated by density matrices that can be mapped to arrowhead matrices of
the from of Eq.~\eqref{extremal_state_matrix}. We now aim to explicitly construct
density matrices (positive semi-definite and unit trace) that map to the extremal arrowhead matrices in Eq.~\eqref{extremal_state_matrix}
and thereby maximise the coherent mismatch. Let us now derive the explicit form of these states in
terms of the decomposition in Eq.~\eqref{rho_decomp_sum} as a weighted sum of the ideal state and an error state
as $\eta \rho_{id} + (1-\eta)  \rho_{err}$.

Let us denote the non-trivial $2$-dimensional block of the density matrix in Eq.~\eqref{extremal_state_matrix} as $M$,
which in the arrowhead representation $\tilde{M} := U M U^\dagger$ yields 
\begin{equation*}
	\tilde{M}:= M_{max} = 		\begin{pmatrix}
		F & C_2  \\
		C_2 & D_2 \\
	\end{pmatrix}.
\end{equation*}
 Let us also assume that $M$ is rank-$2$ (i.e., if it is rank-$1$ than it
represents purely a coherent error and we trivially find that $c=1-F$) which guarantees that
the decomposition in Eq.~\eqref{rho_decomp_sum} exists.
In this case we can uniquely find an optimal $\eta$ from Sec.~\ref{def_state_appendix} for which the difference matrix
$M - \eta | \psi_{id} \rangle \langle \psi_{id} |$
is rank-$1$.  We can thus obtain the following expression for $M$ as a sum of two rank-one
matrices as
\begin{equation} \label{mdecomp}
	M = \delta_1 \rho_{id} +  \delta_2 | \chi \rangle \langle \chi | \quad \text{with}  \quad
	\delta_1:= \eta \quad \text{and} \quad \delta_2 := (1-\eta) \mu_1,
\end{equation}
where $\eta$ is the weight in $\eta \rho_{id} + (1-\eta)  \rho_{err}$
and $\mu_1$ is the largest eigenvalue of the error density matrix  as defined in Definition~\ref{def_state}.
Here the pure state $|\chi \rangle$ can generally be expressed as a linear combination of the first two basis vectors
(and recall that $|\tilde{\psi}_{id}\rangle = (1,0,0, \dots 0 )^T$ and $| \tilde{\phi}_2 \rangle = (0,1,0, \dots 0 )^T$) as
\begin{equation*}
	| \chi \rangle = \sqrt{\alpha}| \psi_{id} \rangle +  \sqrt{1-\alpha}| \phi_2 \rangle
\end{equation*}
for some $\alpha \geq 0$.

We finally obtain the error density matrices that saturate the upper bound of the coherent
mismatch  as
\begin{equation*}
	\rho_{err} = \mu_1 |\chi \rangle \langle \chi | +   R \quad \quad \text{with} \quad \quad R 
	 =
	 \sum_{k=3}^d D_k |\phi_k \rangle \langle \phi_k |,
\end{equation*}
where $R$ is orthogonal to $\langle \psi_{id} |R|\psi_{id} \rangle =  \langle \chi  |R| \chi  \rangle = 0$
as well as we can arbitrarily choose the probability distribution $\{ D_k:  3\leq k \leq d \}$
and we can choose the dominant eigenvalue of $\rho_{err}$ arbitrarily in the range $1/d \leq \mu_1 \leq 1$
as long as $\tr \rho_{err} = \sum_{k=3}^d D_k + \mu_1 = 1 $. As such, any density matrix of the form
$\eta \rho_{id} + (1-\eta)  \rho_{err}$ saturates the corresponding upper bound on the coherent mismatch.\\

\noindent \textbf{Obtaining the upper bound}\\
The upper bound in Eq.~\eqref{general_upper_bound} depends on parameters that
one can only obtain from the arrowhead representation of a quantum state, i.e.,  $D_2$ and $ \lVert \mathcal{C} \rVert^2$.
Let us now derive an alternative upper bound on the coherent mismatch that depends on parameters
of the decomposition $\eta \rho_{id} + (1-\eta)  \rho_{err}$.
We compute this upper bound by exactly computing the coherent mismatch for the extremal density matrices
in Eq.~\eqref{mdecomp}.
 By representing $M$
in the arrowhead basis, we obtain the $2\times2$ block from Eq.~\eqref{mdecomp} as
\begin{equation*}
	\tilde{M} := U M U^\dagger = 
	\left(
	\begin{array}{cc}
		\alpha  \delta_2 + \delta_1 & \delta_2 \sqrt{\alpha (1- \alpha)  } \\
		\delta_2 \sqrt{\alpha (1 - \alpha )  } & \delta_2 (1-\alpha ) \\
	\end{array}
	\right).
\end{equation*}

We can now compute the coherent mismatch $c$ analytically as a function of $\delta_1, \delta_2$
and $\alpha$, and maximise $c$ with respect to $\alpha$. For this reason, we first express the
coherent mismatch analytically as (by computing the first component of the eigenvector as)
\begin{equation*}
	c = -\frac{2 (\alpha -1) \alpha  \delta _2^2}{\delta _1 \left(\sqrt{2 (1-2 \alpha ) \delta _2 \delta _1+\delta _1^2+\delta _2^2}+(2-4 \alpha ) \delta _2\right)+\delta _2 \left((1-2 \alpha ) \sqrt{2 (1-2 \alpha ) \delta _2 \delta _1+\delta _1^2+\delta _2^2}+\delta _2\right)+\delta _1^2}
\end{equation*}

We can find the first order optimality condition by differentiating $c$ with respect to $\alpha$ as
\begin{equation*}
	0 = \frac{\partial c}{\partial \alpha} = \frac{\delta _2^2 \left((1-2 \alpha ) \delta _1+\delta _2\right)}{\left(2 (1-2 \alpha ) \delta _2 \delta _1+\delta _1^2+\delta _2^2\right){}^{3/2}}.
\end{equation*}
We can uniquely solve this equation in terms of $\delta := \delta_2/\delta_1$ from Definition~\ref{def_state} as
\begin{equation*}
	\alpha = \frac{1 + \delta}{2}
\end{equation*}
We finally obtain the coherent mismatch for the above worst-case scenario density matrices as
\begin{equation*}
	c = (1-\sqrt{1-\delta^2} )/2.
\end{equation*}
The above expression is a general upper bound for the coherent mismatch which is saturated by any density matrix that can be mapped
to an arrowhead matrix of the form assumed above.
\\

\noindent \textbf{Remark:}\\
Let us remark that one could use the upper bound that explicitly depends on $\alpha$
\begin{align*}
	c &= \frac{2 (1 - \alpha) \alpha  \delta^2}{
		\left(\sqrt{2 (1-2 \alpha ) \delta _2 \delta _1+\delta _1^2+\delta _2^2}+(2-4 \alpha ) \delta _2\right)/\delta_1
		+
		\delta \left((1-2 \alpha ) \sqrt{2 (1-2 \alpha ) \delta _2 \delta _1+\delta _1^2+\delta _2^2}+\delta _2\right)/\delta_1
		+1}\\
	&\leq 2 (1 - \alpha) \alpha  \delta^2.
\end{align*}
But here we aimed to derive an upper bound that is independent of $\alpha$.

\end{proof}

\section{Number of copies for error suppression \label{proof_copies}} 
The suppression factor $Q$ was introduced in ref~\cite{koczor2020exponential}
which we adapt to the notations used in this work as  $Q:= \lambda_2/\lambda$ (the present work uses a different
convention to denote the eigenvalues of $\rho$ in Eq.~\ref{rho_decomp} when compared to ref.~\cite{koczor2020exponential}).
For small values of $c$ (quadratically scaling region in Fig.~\ref{plot_eigvals}) we can approximate 
$Q = \lambda_2/\lambda  \approx (1-\eta) \mu_1/\eta=  \delta$ and the coherent mismatch
is then bounded as $c \lessapprox Q^2/4$.
It was shown in ref.~\cite{koczor2020exponential} that the number of copies scales logarithmically
with this suppression factor as  $n = (\ln \mathcal{E} + \ln [ \mu_1/2])/\ln Q $, where $\mathcal{E}$ is the target precision
and in the relevant region we approximate $p_{max} \approx \mu_1$ via the dominant eigenvalue of the noise state $\rho_{err}$.
Let us use that $c \lessapprox Q^2/4$ and obtain the approximate bound assuming a target precision $\mathcal{E} \rightarrow 2 \sqrt{c}$ as
$$n \approx \frac{\ln [2 \sqrt{c}] + \ln [ \mu_1/2]}{\ln Q }
= \frac{\ln Q + \ln [ \mu_1/2]}{\ln Q }  
=1 + \frac{ \ln [ \mu_1/2] }{\ln Q }   
\approx 1 + \frac{ \ln [ \mu_1/2] }{\ln (\eta^{-1} -1)\mu_1 }   
$$
and we have used that the suppression factor from ref.~\cite{koczor2020exponential} can be approximated as $Q \approx (\eta^{-1} -1)\mu_1$.
The second term is necessarily positive since $Q < 1$ and applying the ceiling function we obtain that
even in the worst-case scenario one needs at least 2 copies to reach the target precision as
$n \geq 2$. In case if the quantum states are considerably noisy via $(\eta^{-1} -1) > 1/2$ (or equivalently $\eta < 2/3$),
we find that one needs at least 3 copies to reach the target precision. This condition corresponds to a circuit
	error rate $\xi > 0.41$ which is reasonable to assume in practice.

Similarly, in the special case when we set the precision to $2c$, as relevant for eigenstates as discussed in Sec.~\ref{precision_motivation},
we need at least 3 copies as $n \geq 3$.
Furthermore, when $(\eta^{-1} -1) > 1/4$ or equivalently when $\eta < 4/5$ then we need at least $4$ copies
-- and it is reasonable to assume that for practically relevant applications $\lambda \approx \eta < 4/5$ which
	corresponds to a circuit error rate $\xi > 0.22$.

\section{Proof of Statement~\ref{stat_commutator} \label{proof_commutator}}

\begin{lemma}
	The the matrix $C$ in the arrowhead decomposition from Statement~\ref{stat_arrowhead}
	satisfies the following eigenvalue equations
	\begin{equation*}
		C  | u_\pm \rangle = \pm \sigma \, | u_\pm \rangle.
	\end{equation*}
	All vectors $ |v \rangle $ orthogonal to the eigenvectors $ | u_\pm \rangle$ satisfy $ C |v \rangle = 0$.
	Similarly, the commutator $[\rho_{id}, \rho]$ 
	satisfies the eigenvalue equation
	\begin{equation*}
		[\rho_{id}, \rho] | v_\pm \rangle = \pm i \sigma \, | v_\pm \rangle.
	\end{equation*}
	Since the singular values of the two operators are identical, all $p$ norms are equivalent as
	\begin{equation}
	 \lVert [\rho_{id}, \rho] \rVert_p = 	\lVert C \rVert_{p}  = 2^{1/p} \sigma.
	\end{equation}
	The eigenvectors are the following linear combinations
	\begin{equation*}
	| u_\pm \rangle = 	(| \psi_{id} \rangle \pm \frac{1}{\lVert \phi\rVert } | \phi \rangle )/\sqrt{2}
		\quad \quad \text{and} \quad \quad
	| v_\pm \rangle = 	(| \psi_{id} \rangle \pm \frac{i}{\lVert \phi\rVert } | \phi \rangle )/\sqrt{2},
	\end{equation*}
	where the vector $| \phi \rangle$ can be defined implicitly via the commutator as
	$ [\rho_{id}, \rho]	= |\psi_{id} \rangle \langle \phi | - |\phi\rangle \langle \psi_{id} | $,
	and the singular value is $\sigma := \lVert \phi\rVert$.
\end{lemma}

\begin{proof}

\noindent  \textbf{Eigenvalues and norm of $C$:} Let us prove that the matrix $C$ has only two non-zero eigenvalues
$\pm \sigma$. Note that $C$ is a special arrowhead matrix with $D_k = 0$ for all $k$.
Using the expression from Eq.~\eqref{eq_eigenvalue} we can compute the eigenvalues analytically as
\begin{equation}
	P(\sigma) = \sigma - F + \sum_{k=2}^d \frac{C_k^2 }{(D_k - \sigma) } = 0.
\end{equation}

	Here we can use that diagonal entries are $0 = D_k = F 	$. It follows that
	\begin{equation*}
		\sigma - \frac{1}{\sigma} \sum_{k=2} C_k^2 =0.
	\end{equation*}
	Recall that the $p=2$ Hilbert-Schmidt norm can be computed via the sum of
	squares of matrix entries as $\sum_{k=2} C_k^2 =  \lVert C \rVert_{HS}^2 /2$ and
	we thus obtain the following expression for the eigenvalues
	\begin{equation*}
		\lVert C \rVert_{HS}^2 /2 = \sigma^2.
	\end{equation*}
	Indeed, there exist two non-zero solutions $\sigma = \pm \lVert C \rVert_{HS} /\sqrt{2}$.
	Since there are only two non-zero eigenvalues, we can compute the infinity norm as $\lVert C \rVert_{\infty} = \lVert C \rVert_{HS} /\sqrt{2}$.
	In fact, we can compute any $p$-norm of the matrix $C$ as 
	\begin{equation*}
		\lVert C \rVert_{p} = 2^{1/p} \sigma  
	\end{equation*}

\noindent \textbf{Eigenvectors of $C$:} 
	Let us introduce the  vector $|\phi\rangle$ which can be defined via the decomposition
	of the $C$ matrix as  the first row and column vectors of $C$ as
	\begin{equation*}
		C =: |\phi\rangle \langle \psi_{id} | + |\psi_{id} \rangle \langle \phi |.
	\end{equation*}
	Using results of \cite{gu1995divide} we can analytically compute the
	eigenvectors of $C$ using the eigenvalues $\sigma   = \pm \lVert C \rVert_{HS} /\sqrt{2}$
	from statement 1 as
	\begin{equation*}
		| u_\pm \rangle = (\frac{1}{\sqrt{2}}, \pm \frac{C_1}{\lVert C \rVert_{HS}},\pm \frac{C_2}{\lVert C \rVert_{HS}}, \cdots \pm\frac{C_d}{\lVert C \rVert_{HS}})^T
		= (| \psi_{id} \rangle \pm \frac{1}{\lVert \phi\rVert } | \phi \rangle )/\sqrt{2} ,
	\end{equation*}
	where $\lVert \phi\rVert = \sigma   = \lVert C \rVert_{HS} /\sqrt{2}$.
	Indeed we can confirm that the eigenvalue equation is satisfied as 
	\begin{equation*}
		C | u_\pm \rangle = ( |\phi\rangle \langle \psi_{id} | + |\psi_{id} \rangle \langle \phi | )
		(| \psi_{id} \rangle \pm \frac{1}{\lVert \phi\rVert } | \phi \rangle )/\sqrt{2}
		= 
		\pm \lVert \phi\rVert \, | u_\pm \rangle.
	\end{equation*}
\noindent \textbf{Eigenvalues and norm of the commutator:} 	
	We can similarly write the commutator as
	\begin{equation*}
		 [\rho_{id}, \rho]	= |\psi_{id} \rangle \langle \phi | - |\phi\rangle \langle \psi_{id} |,
	\end{equation*}
	and its eigenvectors are
	\begin{equation*}
		| v_\pm \rangle 
		= (| \psi_{id} \rangle \pm  \frac{i}{\lVert \phi\rVert } | \phi \rangle )/\sqrt{2} ,
	\end{equation*}
	and indeed its eigenvalues are $\pm i \lVert C \rVert_{HS} /\sqrt{2}=\pm i \lVert \phi\rVert = \pm i \sigma$ via the eigenvalue equation
	\begin{equation*}
		[\rho_{id}, \rho] | u_\pm \rangle = ( |\psi_{id} \rangle \langle \phi | -  |\phi\rangle \langle \psi_{id} |)
		(| \psi_{id} \rangle \pm \frac{i}{\lVert \phi\rVert } | \phi \rangle )/\sqrt{2}
		= 
		\pm i \lVert \phi\rVert \, | u_\pm \rangle.
	\end{equation*}
	\textbf{Equivalence of norms:} We have shown in the previous statement that $C$ and the commutator 
	share the same singular values. It immediately follows with denoting the singular value
	$\sigma = \lVert C \rVert_{HS}/\sqrt{2}$ that the norms are equivalent as
	\begin{equation*}
		\lVert C \rVert_{p} = 2^{1/p} \sigma = \lVert [\rho_{id}, \rho] \rVert_p.
	\end{equation*}
\end{proof}

\begin{lemma}
	The Hilbert-Schmidt norm of the commutator can be computed exactly as
	\begin{equation*}
		\frac{1}{2}\lVert  [\rho_{id}, \rho] \rVert_{HS}^2  =:	\sigma^2 =
		\var[\rho] =   \langle \rho^2 \rangle -  \langle \rho \rangle ^2 = \langle  \psi_{id} | \rho^2 | \psi_{id} \rangle  - F^2
	\end{equation*}
\end{lemma}
\begin{proof}
The Hilbert-Schmidt norm is computed via the trace
\begin{align*}
	\lVert  [\rho_{id}, \rho] \rVert_{HS}^2  
	=& \tr\{ [\rho_{id}, \rho] [\rho, \rho_{id}]\} 
	= \tr\{ (\rho_{id} \, \rho - \rho \,  \rho_{id})  ( \rho \,  \rho_{id} - \rho_{id} \, \rho)  \} \\
	=& -\tr[ \rho_{id} \, \rho \,  \rho_{id} \, \rho   ]
	-\tr[\rho \,  \rho_{id} \,  \rho \,  \rho_{id}   ]
	+ \tr[\rho_{id} \, \rho \,  \rho \,  \rho_{id}   ]
	+ \tr[\rho \,  \rho_{id}  \, \rho_{id} \, \rho   ],
\end{align*}
here we can simplify the expressions using that $\rho_{id} \,  \rho \,  \rho_{id} = F \rho_{id}$
and we can also use the cyclic reordering property of the trace 
so we obtain 
\begin{equation*}
	\sigma^2 = 
	\lVert  [\rho_{id}, \rho] \rVert_{HS} ^2 /2
	=
	\tr[\rho_{id} \, \rho \,  \rho \,  \rho_{id}   ]
	- F \tr[\rho \,  \rho_{id}   ]
	= \langle \psi_{id} | \rho^2 | \psi_{id} \rangle - \langle \psi_{id} | \rho | \psi_{id} \rangle^2,
\end{equation*}
where $\sigma$ is the common, only non-zero singular value with $C$
and we have also used that $F=\langle \psi_{id} | \rho | \psi_{id} \rangle$.

This is indeed a quantum-mechanical variance and using elementary statistics
\begin{equation*}
	\sigma^2 = \var[\rho] = \langle \rho^2 \rangle -  \langle \rho \rangle ^2 =  \langle (\rho - F)^2 \rangle ,
\end{equation*}
where $\langle X \rangle := \langle \psi_{id} | X | \psi_{id} \rangle$.

Our final result is that we can compute the singular value via the above variance
\begin{equation*}
	\sigma^2 = \var[\rho]=   \langle \rho^2 \rangle -  \langle \rho \rangle ^2
	= \langle  \psi_{id} | (\rho - F)^2 | \psi_{id} \rangle 
	= \langle  \psi_{id} | (\rho - F \rho_{id})^2 | \psi_{id} \rangle 
\end{equation*}

The norm $\lVert  [\rho_{id}, \rho] \rVert_{HS}^2  $ of the commutator is given by the variance
\begin{equation*}
	\sigma^2 = \var[\rho]=   \langle \rho^2 \rangle -  \langle \rho \rangle ^2 = \langle  \psi_{id} | \rho^2 | \psi_{id} \rangle  - F^2
\end{equation*}

\end{proof}

\section{Proof of Theorem~\ref{theo_commut_upb}: upper bound in terms of the commutator \label{proof_upb_commutator}}
Recall that density matrices that can be mapped to arrowhead matrices of the form of Eq.~\eqref{extremal_state_matrix}
	as
\begin{equation}
	A_{max} :=
	\begin{pmatrix}
		F & C_2 & 0 & \dots & 0 \\
		C_2 & D_2  &  &  &  \\
		0 &  & D_3 \\
		\vdots & & & \ddots &\vdots \\
		0 &  &  & \dots & D_d
	\end{pmatrix}
\end{equation}
saturate the upper bound of the coherent mismatch in Eq.~\eqref{general_upper_bound}.
Here the 2-dimensional sub-block $M:= M_{max} = 		\begin{pmatrix}
	F & C_2  \\
	C_2 & D_2 \\
\end{pmatrix}$ is a $2$-dimensional arrowhead matrix and we can write
the corresponding coherent mismatch using Statement~\ref{stat_coherent} as
\begin{equation}
	c(\Xi) = 1- |\langle \psi_{id} | \psi \rangle|^2 = 1- [1 + \Xi ]^{-1},
\end{equation}	
where the term $\Xi$ yields the simplified expression
	\begin{equation}
		\Xi =   \frac{C_2^2 }{(D_2 - \lambda)^2 } =  \frac{\sigma^2 }{(D_2 - \lambda)^2 },
	\end{equation}
	where we have used that $\sigma \equiv C_2$.
	Let us analytically express $D_2$ in terms of the eigenvalues $\lambda$ and $\lambda'$
	of the $2$-dimensional matrix, and in terms of $C_2$.
	We can analytically solve these eigenvalues as
	\begin{align*}
		\lambda &=   \frac{1}{2} [D_2 + F +\sqrt{4 C_2^2  +  (D_2 - F)^2}  ] \\
		\lambda'&= \frac{1}{2} [D_2 + F -\sqrt{4 C_2^2  +  (D_2 - F)^2}  ]  
	\end{align*}
	and express $F$ and $D_2$ in terms of the eigenvalues $\lambda$ and $\lambda'$ as
	\begin{align*}
		D_2 = & \frac{ C_2^2 }{ F-\lambda } + \lambda\\
		F  =  & \frac{C_2^2 + D_2 \lambda' - \lambda'^2 }{D_2  -  \lambda'}
	\end{align*}
	We can now express $\Xi$ either in terms of $(\lambda,F)$ or in terms of $(\lambda,\lambda')$ and
	we now substitute $C_2 \equiv \sigma$ as
	\begin{align}
		\Xi(\lambda,F) &= \frac{ (F-\lambda )^2 }{\sigma^2},  \\
		\Xi(\lambda,\lambda') &= \frac{2 \sigma^2}{  ( \lambda -\lambda' )
			\sqrt{ ( \lambda -\lambda' )^2 - 4 \sigma^2}
			+ ( \lambda -\lambda' )^2  - 2\sigma^2}.
	\end{align}
	From the first equation we can see that $\Xi$ ultimately depends on the ratio
	of the gap $\lambda-F$ and the commutator $\sigma$.
	Similarly, the second equation only depends on $\sigma$ and on the gap between the two
	eigenvalues $\lambda-\lambda'$. Let us remark that $\lambda'$ is the second
	largest eigenvalue of the density matrix as $\lambda_2 \equiv \lambda'$ in case if
	the diagonal entries of the extremal arrowhead matrix
	from Eq.~\eqref{extremal_state_matrix} are such that $\lambda' \geq D_3$.
	This condition is generally satisfied when $D_2\geq D_3 + C_2^2/(F-D_3)$.
	Nevertheless, without loss of generality we can assume in the following that
	$\lambda_2 \equiv \lambda'$, in which case the resulting upper bound will only
	be saturated by arrowhead matrices $A_{max}$ which satisfy $D_2\geq D_3 + C_2^2/(F-D_3)$.

	We can simplify our expression for $\Xi(\lambda,\lambda_2)$
	by introducing the factor $\Delta := \sigma/( \lambda -\lambda _2 )$.
	This allows us to directly express $\Xi$ in terms of $\Delta$ via the above equation as
	\begin{equation}
		\Xi(\Delta) = \frac{2 \Delta^2}{ 1 - 2 \Delta^2 + \sqrt{1-4 \Delta^2}  }.
	\end{equation}
	It is immediately clear that the term in $\sqrt{1-4 \Delta^2}$ is non-negative when $1-4 \Delta^2 \geq 0$ and
	thus our bound holds when $1/2 \geq \Delta$.
	Let us finally write $\Delta$ in terms of $\lambda$ and $Q:=\lambda_2/\lambda$ and in terms
	of $\sigma_r:= \sigma/\lambda$ as
	\begin{equation*}
		\Delta = \sigma/(\lambda - \lambda_2) = \sigma/(\lambda - Q\lambda) = \frac{\sigma}{ \lambda(1- Q) } =  \frac{\sigma_r}{1- Q }
	\end{equation*}

	We can finally simplify the expression for $c$ and find the surprisingly simplified formula as
	\begin{equation*}
		c= 1-[1+\Xi(\Delta)]^{-1} = (1- \sqrt{1- 4\Delta^2})/2.
	\end{equation*}
	We can expand this for small $\Delta$ and find that indeed the coherent mismatch scales quadratically
	with the commutator norm as
	\begin{equation*}
		c=  (1- \sqrt{1- 4\Delta^2})/2 = \Delta^2	+ \Delta^4 +\mathcal{O}(\Delta^6).
	\end{equation*}
The above equations express the coherent mismatch for the extremal states and thus guarantee a general upper bound for $c$.
This bound is saturated by density matrices that can be mapped to arrowhead matrices of the form of
Eq.~\eqref{extremal_state_matrix} with the additional constraint $D_2\geq D_3 + C_2^2/(F-D_3)$
that ensures that the smaller eigenvalue $\lambda'$ of the two-dimensional matrix block
$M_{max}$ is the second largest eigenvalue of the arrowhead matrix.

\section{Proof of Lemma~\ref{theo_commut_lowb}: lower bound in terms of the commutator \label{proof_lowb_commutator}}
\begin{proof}
Recall that density matrices that can be mapped to arrowhead matrices of the form of Eq.~\eqref{extremal_state_matrix_lower}
	as
	\begin{equation}
		A_{min} :=
		\begin{pmatrix}
			F & C_m & 0 & \dots & 0 \\
			C_m & D_m  &  &  &  \\
			0 &  & D_2 \\
			\vdots & & & \ddots &\vdots \\
			0 &  &  & \dots & D_d
		\end{pmatrix}
	\end{equation}
	saturate the lower bound of the coherent mismatch in Eq.~\eqref{general_upper_bound}.
	Here $D_m$ is the smallest non-zero diagonal entry in the arrowhead matrix.
	Here the 2-dimensional sub-block $M:= M_{min} = 		\begin{pmatrix}
		F & C_m  \\
		C_m & D_m \\
	\end{pmatrix}$ is a $2$-dimensional arrowhead matrix and we can compute
	the corresponding coherent mismatch similarly as for the upper bound
	in Appendix~\ref{proof_upb_commutator}.

For this reason, let us introduce $\dmin := \sigma_r/(1 -\qmin)$ where $\sigma_r := \sigma/\lambda$ and $\qmin:= \lambda_m/\lambda$ is
the ratio of the smallest and largest eigenvalues.
With this, we can compute the analytical expression for $c$ and obtain the expression
for the coherent mismatch as
\begin{equation*}
	c= 1-[1+\Xi(\Delta_\textrm{min})]^{-1} = (1- \sqrt{1- 4\Delta_\textrm{min}^2})/2,
\end{equation*}
Let us remark that we can expand the above expression for small $\Delta_\textrm{min}$ as
\begin{equation*}
	c=  (1- \sqrt{1- 4\Delta^2})/2 = \Delta_\textrm{min}^2	+ \Delta_\textrm{min}^4 +\mathcal{O}(\Delta_\textrm{min}^6).
\end{equation*}
The above equations express the coherent mismatch for the extremal states and thus guarantee a general lower bound for $c$.
We note that the eigenvalues of the 2-dimensional matrix $M_{min}$ are guaranteed to be the
largest and the smallest non-zero eigenvalues of the density matrix due to the interlacing property.
It follows that the above lower bound is saturated by any density matrix that can be mapped to an
arrowhead matrix of the form of $A_{min}$ above.
\end{proof}

\section{Commutators in noisy quantum circuits \label{sec_ciruit_commutator}}
\noindent \textbf{The noise model}\\
Here we assume a noise channel that maps to a noisy state via $\nu$ noisy gates as introduced in Eq.~\eqref{error_channel}.
This noisy quantum circuit is in the form of a product of noisy quantum gates as
\begin{equation*}
	\rho := \Phi_\nu \Phi_{\nu-1} \cdots \Phi_{1} \, \rho_{0},
\end{equation*}
where every gate can be written in terms of the Kraus map from Eq.~\eqref{error_channel} as
\begin{equation*}
	\Phi_{k} \,  \rho := (1-\epsilon) U_k \rho U_k^\dagger + \epsilon \sum_{j=1}^{K} M_{j k} \rho M_{jk}^\dagger.
\end{equation*}
In the following we focus on the special case of $K=1$ for ease of notation. A prominent example is the dephasing noise channel
in which case $M_k = Z_k U_k$, where $Z_k$ is a Pauli Z operator that acts on the same qubit(s) as the unitary $U_k$.

This family of noise models can be understood via the analogy to flipping $\nu$ coins: every coin has a probability 
$\epsilon$ to yield heads (error event via $M_k$) and probability $1-\epsilon$ to yield tails (no error via $U_k$).
The probability that no error happens throughout the entire circuit (all tails) is then $(1-\epsilon)^\nu$.
This allows  us to write the resulting density matrix into the following form
\begin{equation*}
	\rho = p_0 \rho_{id} + \epsilon (1-\epsilon)^{\nu-1} \rho_1  + \epsilon (1-\epsilon)^{\nu-1} \rho_2 + \dots \epsilon^2 (1-\epsilon)^{\nu-2} \rho_{12} +  \dots  \epsilon^\nu \rho_{1234 \dots \nu}.
\end{equation*}
Here every term represents a  pure state, for example
$\rho_1  = | \mathcal{E}_1 \rangle \langle \mathcal{E}_1 |$ is a pure state in which
an error occurred at gate $1$ and therefore its state vector can be expressed as
\begin{equation*}
	| \mathcal{E}_1 \rangle := U |\underline{0}\rangle = U_\nu U_{\nu-1} \cdots U_2 M_1  |\underline{0}\rangle,
\end{equation*}
which happens with probability  $\epsilon (1-\epsilon)^{\nu-1}$.
Similarly, $\rho_{12}$ is the pure state in which errors occurred at gate 1 and 2.
There are overall $1 + \sum_{m=1}^\nu \binom{\nu}{m} = 2^\nu$ terms in the above sum
and we can compute their probabilities as
\begin{align*}
	p_0 &= (1-\epsilon)^\nu \quad &\text{probability of no error}  \\
	p_1 &= p_2 = \dots p_\nu = \epsilon (1-\epsilon)^{\nu-1} \quad & \text{probability of a signle error}  \\
	p_{12} & = p_{13} = \dots p_{\nu-1,\nu} = \epsilon^2 (1-\epsilon)^{\nu-2} \quad &\text{probability of a double error}\\
	 & \vdots \\
	p_{12\dots \nu} &= \epsilon^\nu & \text{probability of $\nu$ errors}
\end{align*}

\noindent \textbf{Expressing the commutator}\\
Let us now compute the commutator $[\rho_{id},\rho]$  with $\rho_{id} := | \psi_{id} \rangle \langle \psi_{id} |$ explicitly 
(using that the first term cancels out since it commutes with $\rho_{id}$) as
\begin{align*}
	[\rho_{id},\rho] =   \, &p_1 | \psi_{id} \rangle \langle \psi_{id} |  | \mathcal{E}_1 \rangle \langle \mathcal{E}_1 | 
	- p_1  | \mathcal{E}_1 \rangle \langle \mathcal{E}_1 | | \psi_{id} \rangle \langle \psi_{id} |\\
	+& p_2 | \psi_{id} \rangle \langle \psi_{id} |  | \mathcal{E}_2 \rangle \langle \mathcal{E}_2 | 
	- p_2  | \mathcal{E}_2 \rangle \langle \mathcal{E}_2 | | \psi_{id} \rangle \langle \psi_{id} |
	+\dots\\
\end{align*}
Let us introduce the coefficients, for example $c_1:= p_1 \langle \mathcal{E}_1  |  \psi_{id}   \rangle$,
and let us write the commutator in terms of these coefficients as
\begin{align*}
	[\rho_{id},\rho]	= & c_1^*  | \psi_{id} \rangle \langle \mathcal{E}_1 | 
	+ c_2^*  | \psi_{id} \rangle \langle \mathcal{E}_2 | \dots  - c_1  |\mathcal{E}_1 \rangle \langle   \psi_{id}|  -  c_2 |  \mathcal{E}_1 \rangle \langle \psi_{id} |\\
	=&| \psi_{id} \rangle \left(  c_1^* \langle \mathcal{E}_1 | + c_2^* \langle \mathcal{E}_2 | + \dots \right) 
	-
	 \left(  c_1  |\mathcal{E}_1 \rangle +  c_2  |\mathcal{E}_2 \rangle  \right) \langle\psi_{id}|.
\end{align*}
Indeed, analogously to the proofs in Appendix~\ref{proof_commutator} we find the arrowhead structure of the commutator as
$ [\rho_{id},\rho] =  | \psi_{id} \rangle \langle \phi | - |\phi \rangle \langle\psi_{id}| $,
for which expression we can introduce the vector
\begin{equation*}
	|\phi \rangle :=   c_1  | \mathcal{E}_1 \rangle +  c_2  | \mathcal{E}_2 \rangle + \dots c_{123\dots\nu} |\mathcal{E}_{123 \dots \nu} \rangle.
\end{equation*}
Let us now compute the Hilbert-Schmidt norm of the commutator via
\begin{equation*}
	\lVert [\rho_{id},\rho] \rVert_{HS}^2  
	=
	\tr[ (| \psi_{id} \rangle \langle \phi | - |\phi \rangle \langle\psi_{id}|) (|\phi \rangle \langle\psi_{id}| - | \psi_{id} \rangle \langle \phi |)]
	=
	2  \lVert \phi \rVert^2  - 2 |\langle \phi | \psi_{id} \rangle|^2.
\end{equation*}
We thus conclude that the commutator from Statement~\ref{stat_commutator} can be computed via
\begin{equation*}
 \sigma^2 = 	\lVert \phi \rVert^2  -  |\langle \phi | \psi_{id} \rangle|^2.
\end{equation*}
Let us express the vector norm by introducing the index set $I = \{ 1,2, \dots 2^\nu-1 \}$ whose elements $\textbf{k}\in I$ index the individual error events.
The vector norm is then given by a sum over these events as
\begin{equation*}
	 \lVert \phi \rVert^2 = \sum_{\textbf{k}, \textbf{l}\in I} c_\textbf{k}^* c_\textbf{l}  \langle \mathcal{E}_\textbf{k} | \mathcal{E}_\textbf{l} \rangle
	  = \sum_{\textbf{k}, \textbf{l} \in I} p_\textbf{k} p_\textbf{l}  \langle \mathcal{E}_\textbf{k} | \psi_{id} \rangle \langle \psi_{id} |\mathcal{E}_\textbf{l} \rangle
	   \langle \mathcal{E}_\textbf{k} | \mathcal{E}_\textbf{l} \rangle.
\end{equation*}
Similarly, we can express the overlap via the summation
\begin{equation*}
	|\langle \psi_{id}  | \phi  \rangle|^2 =   | \sum_{\textbf{k} \in I}   c_\textbf{k}  \langle \psi_{id} | \mathcal{E}_\textbf{k} \rangle  |^2
	=  | \sum_{\textbf{k} \in I}   p_\textbf{k}  |\langle \psi_{id} | \mathcal{E}_\textbf{k} \rangle|^2  |^2
	=   \sum_{\textbf{k},\textbf{l} \in I}   p_\textbf{k} p_\textbf{l} 
	 |\langle \psi_{id} | \mathcal{E}_\textbf{k} \rangle|^2  \, |\langle \psi_{id} | \mathcal{E}_\textbf{l} \rangle|^2 .
\end{equation*}
Let us introduce the notation $o_\textbf{k} :=  \langle \mathcal{E}_\textbf{k} | \psi_{id} \rangle$
and write that
\begin{equation*}
	\lVert \phi \rVert^2 - 	|\langle \psi_{id}  | \phi  \rangle|^2
	=
	 \sum_{\textbf{k},\textbf{l} \in I}   p_\textbf{k} p_\textbf{l} 
	 [ o_\textbf{k} o_\textbf{l}^* \langle \mathcal{E}_\textbf{k} | \mathcal{E}_\textbf{l} \rangle - |o_\textbf{k}|^2 \,  |o_\textbf{l}|^2 ]
	 	=
	 \sum_{\textbf{k},\textbf{l} \in I}   p_\textbf{k} p_\textbf{l} 
	 [ o_\textbf{k} o_\textbf{l}^* ( \langle \mathcal{E}_\textbf{k} | \mathcal{E}_\textbf{l} \rangle - o_\textbf{k}^*   o_\textbf{l} )]	.
\end{equation*}

Let us finally introduce the notation $\mathcal{L}_{\textbf{k} \textbf{l}}:= 	\mathrm{Re}[  o_\textbf{k} o_\textbf{l}^* ( \langle \mathcal{E}_\textbf{k} | \mathcal{E}_\textbf{l} \rangle - o_\textbf{k}^*   o_\textbf{l} ) ]$,
which results in the compact expression
\begin{equation}\label{phinorm}
	\lVert \phi \rVert^2 - 	|\langle \psi_{id}  | \phi  \rangle|^2
	=
	\sum_{\textbf{k}\in I}   p_\textbf{k}^2 \mathcal{L}_{\textbf{k} \textbf{k}}
	+
	2 \sum_{\textbf{k} < \textbf{l} \in I}   p_\textbf{k} p_\textbf{l} 
	\mathcal{L}_{\textbf{k} \textbf{l}} 	.
\end{equation}

\noindent \textbf{Simplifying the scalar products}\\
Let us express the noise states in terms of a parallel and orthogonal component to the ideal state as
\begin{equation*}
	| \mathcal{E}_\textbf{k} \rangle =  a_\textbf{k} | \psi_{id} \rangle + b_\textbf{k} 	| \Psi_\textbf{k} \rangle
\end{equation*}
for some $| \Psi_\textbf{k} \rangle$ which is orthogonal to $| \psi_{id} \rangle$ and we are free to choose
the complex phase of $a_\textbf{k}$ and $b_\textbf{k}$ which allows us to define them to be real and non-negative. 
It follows that $ b_\textbf{k}  =  \sqrt{1-  a_\textbf{k}^2 }$,
and we can express the scalar products as
\begin{equation*}
	o_\textbf{k} :=  \langle \mathcal{E}_\textbf{k} | \psi_{id} \rangle = a_\textbf{k},
	\quad \quad \quad
	\langle \mathcal{E}_\textbf{k} | \mathcal{E}_\textbf{l} \rangle = a_\textbf{k} a_\textbf{l} + b_\textbf{k} b_\textbf{l} \langle \Psi_\textbf{k}  | \Psi_\textbf{l}, \rangle.
\end{equation*}
We finally obtain the convenient expression for the terms $ \mathcal{L}_{\textbf{k} \textbf{l}}$ in the summation in  Eq.~\eqref{phinorm} as
\begin{equation*}
	 \mathcal{L}_{\textbf{k} \textbf{l}} =
	 \mathrm{Re}[  o_\textbf{k} o_\textbf{l}^* ( \langle \mathcal{E}_\textbf{k} | \mathcal{E}_\textbf{l} \rangle - o_\textbf{k}^*   o_\textbf{l} )]
	 =
	 a_\textbf{k} a_\textbf{l}  b_\textbf{k} b_\textbf{l} \mathrm{Re}[  \langle \Psi_\textbf{k}  | \Psi_\textbf{l} \rangle] ,
\end{equation*}
where $a_\textbf{k} a_\textbf{l}  b_\textbf{k} b_\textbf{l}$ are real, non-negative while the terms
 $ -1\leq \mathrm{Re}[  \langle \Psi_\textbf{k}  | \Psi_\textbf{l} \rangle] \leq 1$ express the overlaps
between the different error states in a phase-sensitive manner -- these terms depend on the complex phase
angle between two error states. Note that we have already fixed a complex phase gauge relative to the
ideal state when we chose real, non-negative $a_\textbf{k}$ and $b_\textbf{k}$ and thus the complex phase of
$ \langle \Psi_\textbf{k}  | \Psi_\textbf{l} \rangle$ has no further complex phase gauge freedom.

\noindent\textbf{Diagonal terms and general upper bound:}\\
Every term in the first summation in  Eq.~\eqref{phinorm} is is strictly non-negative 
$ 0 \leq \mathcal{L}_{\textbf{k} \textbf{k}}  \leq 1$
and is generally upper bounded (since we can express it as $a_\textbf{k}^2 (1- a_\textbf{k}^2)$).
We can evaluate the upper bound analytically using an identity of the binomial coefficients as
\begin{equation}\label{eq_prob_summation}
	\sum_{\textbf{k}\in I} p_\textbf{k}^2 \, 	\mathcal{L}_{\textbf{k} \textbf{k}}
	\leq
	\sum_{\textbf{k}\in I} p_\textbf{k}^2 =	 	\sum_{k = 1}^\nu  \binom{\nu}{k} [(1-\epsilon)^{n-k}\epsilon^{k}]^2 
=
(1-\epsilon )^{2 \nu} \left(\left(\frac{2 (\epsilon -1) \epsilon +1}{(\epsilon -1)^2}\right)^\nu-1\right),
\end{equation}
where the last equality is the analyitcal evaluation of the summation. We will analyse this solution later.

\noindent\textbf{Off-diagonal terms and general upper bound:}\\
The second term in Eq.~\eqref{phinorm} is a sum over exponentially many $\mathcal{O}(|I|^2) = \mathcal{O}(2^{2\nu})$ terms that
depend on the relative phase between the different error states via the scalar products $ \langle \Psi_\textbf{k}  | \Psi_\textbf{l} \rangle$, 
in which the global complex phase of the states $| \Psi_\textbf{k} \rangle$ have been fixed.
 Let us introduce the notation
\begin{equation*}
	\Xi := 	2\sum_{\textbf{k} < \textbf{l} \in I}   p_\textbf{k} p_\textbf{l} 	 \mathcal{L}_{\textbf{k} \textbf{l}}  .
\end{equation*}
In general we can upper bound $\Xi$ using that $|\mathcal{L}_{\textbf{k} \textbf{l}} | \leq 1$ as
\begin{equation*}
| \Xi| \leq 	\lVert \phi \rVert^2 - 	|\langle \psi_{id}  | \phi  \rangle|^2  \leq \sum_{\textbf{k} , \textbf{l}\in I} p_\textbf{k} p_\textbf{l} 
	=
	 \sum_{\textbf{k} \in I} p_\textbf{k}\sum_{\textbf{k} \in I} p_\textbf{l}
	=
	[\sum_{k = 1}^\nu  \binom{\nu}{k} (1-\epsilon)^{n-k}\epsilon^{k}]^2
	=
	[1 - (1-\epsilon)^\nu]^2
	= (1-\tilde{\eta})^2
\end{equation*}
where $ \tilde{\eta}$ was defined in Eq.~\eqref{eq_eta_tilde} as the probability that no error happens.
This upper bound is general, however, we omitted the signs of the summands and the bound is thus going to be very pessimistic
as discussed in the main text.

\noindent \textbf{Off-diagonal terms as random variables:}\\
Let us make the following, rather artificial assumption: the terms $\mathcal{L}_{\textbf{k} \textbf{l}} $
are independent random variables with mean $\langle \mathcal{L}_{\textbf{k} \textbf{l}} \rangle = 0$ and some variance
$s_{\textbf{k} \textbf{l}}^2$, i.e., it is equally likely that they are positive or negative. Due to the bound
$|\mathcal{L}_{\textbf{k} \textbf{l}} | \leq 1$ from the previous subsection, the variance of these
random variables is bounded as $s_{\textbf{k} \textbf{l}}^2:= \var[  \mathcal{L}_{\textbf{k} \textbf{l}}   ] \leq 1$.

Note that the total variance $s^2_{tot}$ of the full sum grows proportionally with the number of terms:
Recall that the total variance in a linear model is expressed via the formula 
\begin{equation*}
\var[ 
	\sum_{\textbf{k},\textbf{l} \in I}   p_\textbf{k} p_\textbf{l} 
	\mathcal{L}_{\textbf{k} \textbf{l}}	
	]
	=
	\sum_{\textbf{k},\textbf{l} \in I}   p_\textbf{k}^2 p_\textbf{l}^2
	s_{\textbf{k} \textbf{l}}^2
	\leq
	s^2		\sum_{\textbf{k},\textbf{l} \in I}   p_\textbf{k}^2 p_\textbf{l}^2
	= s^2		[\sum_{\textbf{k} \in I}   p_\textbf{k}^2]^2.
\end{equation*}

We have also introduced the notation $1 \geq s \geq s_{\textbf{k} \textbf{l}}$ to denote a global upper bound of the variances.
In complete generality we can state that any instance will be very likely to be upper bounded by a multiple of the
standard deviation of the distribution (square-root of the variance). For example,  in case of a normal distribution $3$-times the standard deviation corresponds
to a confidence level above $99\%$. We denote the constant as $s'$ that sets the confidence level
relative to the square-root of the variance $s$  (e.g., $s'=3s$) and obtain the upper bound with high confidence as
\begin{equation*}
	\lVert \phi \rVert^2 - 	|\langle \psi_{id}  | \phi  \rangle|^2 \leq 
	s' \sum_{\textbf{k} \in I}   p_\textbf{k}^2
	=s' f
\end{equation*}
Recall that we have already evaluated this summation analytically (see diagonal entries) and obtained the expression for $f$.
Let us now analyse the solution.

%-------------------------
\begin{figure*}[tb]
	\begin{centering}
		\includegraphics[width=\textwidth]{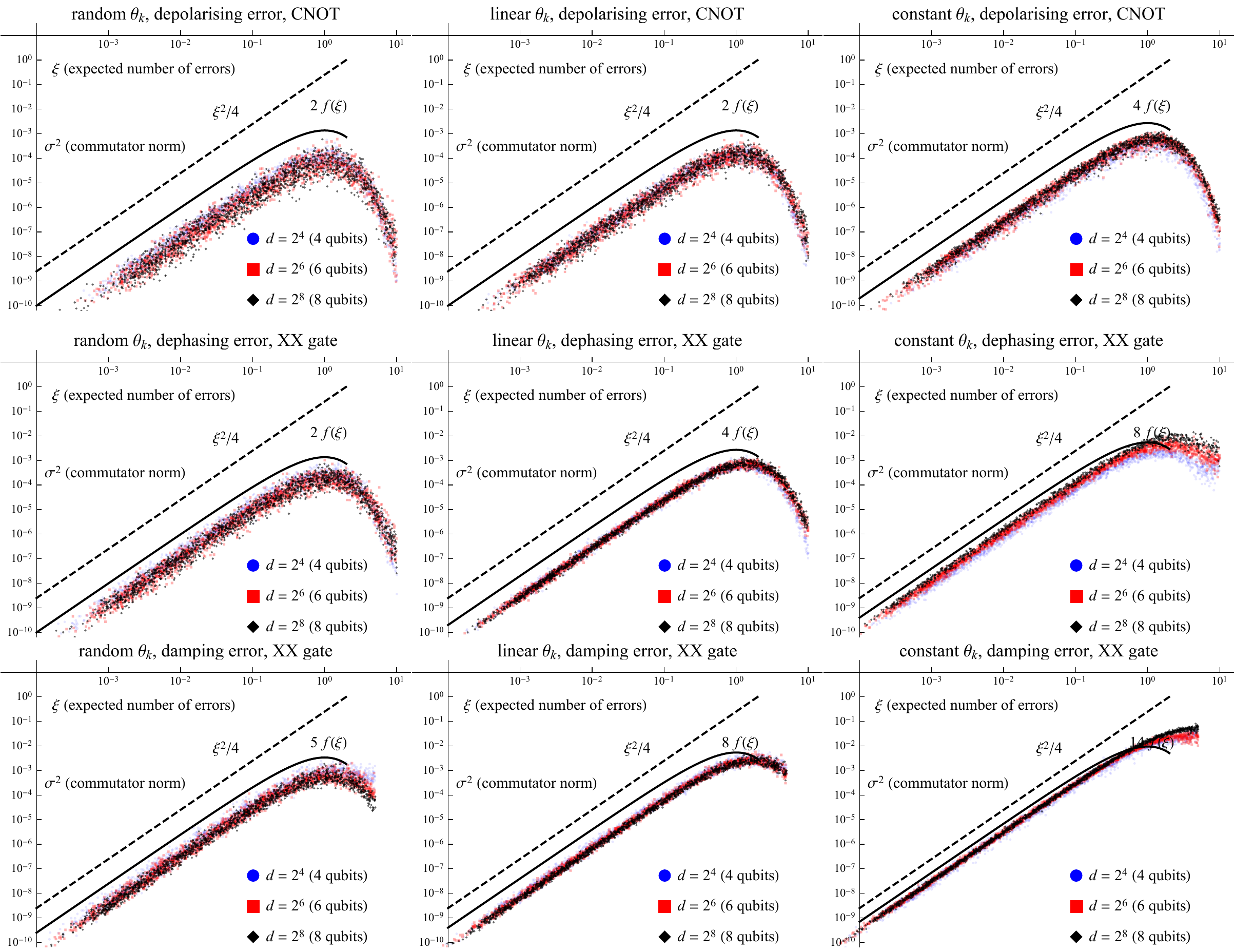}
		\caption{
			Simulated circuits similar to the ones in Fig.~\ref{plot_noisemodel}
			with $\nu=200$ gates under different noise models.
			Uniformly randomly generated rotation angles $\theta_k \in (-\pi, \pi)$ (first column), rotation angles
			increase linearly $\theta_k = 0.01 k$ (second column)
			and constant rotation angles $\theta_k = 0.2$ (third column).
			Notice that the numerical data seems to be slightly shifted to the right when compared to the
			upper bounds (solid black lines).
			This is because a portion of the gates in the simulated circuits
			commute with the noise Kraus maps as discussed in Appendix~\ref{app_general_kraus}.
			\label{plot_noisemodel_table}
		}
	\end{centering}
\end{figure*}
%-------------------------

\subsection{Analysing the solution \label{app_f_analyse}}
Let us now analyse the upper bound function that we have obtained above as
\begin{equation*}
 f :=	(1-\epsilon )^{2 \nu} \left(\left(\frac{2 (\epsilon -1) \epsilon +1}{(\epsilon -1)^2}\right)^\nu-1\right).
\end{equation*}

The first term $ (1-\epsilon )^{2 \nu} =: \tilde{\eta}^2$ in the solution is identical to the square of exponential decay of
the incoherent fidelity, i.e., probability that no error happens in Eq.~\eqref{eq_eta_tilde}.
We have approximated this probability for large $\nu$ and for a varying circuit error rates $\xi := \epsilon \nu$
as $(1-\epsilon )^{2 \nu} = (1- \xi /\nu )^{2 \nu} \approx e^{-2\xi}$. The second term in the solution
can be simplified as
\begin{equation*}
	\left( \left(\frac{2 (\epsilon -1) \epsilon +1}{(\epsilon -1)^2}\right)^\nu-1\right) 
	=
	\exp[ \nu \ln \frac{2 (\epsilon -1) \epsilon +1}{(\epsilon -1)^2}  ] - 1
	=
	e^{\nu (\epsilon^2 + 2 \epsilon^3 + \dots)} - 1
	=
	e^{\xi^2 / \nu  + 2 \xi^3/\nu^2 + \dots} - 1
	\approx
	e^{\xi^2 / \nu } - 1
\end{equation*}
and we have used that for  bounded $\xi$ and large $\nu$ all higher order terms can be neglected
and keep only the leading term $\xi^2/\nu$. Combining the two approximations we finally obtain the 
approximate upper bound as 
\begin{equation*}
	f \approx 	e^{-2\xi} (e^{\xi^2/\nu } - 1) 
	\approx  e^{-2\xi} \, \xi^2/\nu,
\end{equation*}
which approximation has an additive error that scales with $(\xi^2/\nu)^2 \ll 1$.

The function $f(\xi)$ can be divided into 3 distinct regions:
for $\xi \ll 1$ it grows quadratically for fixed $\nu$ as $f(\xi) \approx  \xi^2/\nu $,
or equivalently it grows linearly for fixed $\epsilon$ as $f(\xi) \approx \epsilon \xi $.
The function then reaches its maximum $f(\xi_{max}) $ at around $\xi \approx 1/2$ due to the expansion
\begin{equation}
	\xi_{max} = -\frac{\log \frac{2-\epsilon }{2}}{\epsilon } = 1/2 + \epsilon/8 + \dots
\end{equation}
It is also interesting to note that global maximum of the function 
\begin{equation*}
	f(\xi_{max})	=  \frac{\left(1-\frac{\epsilon }{2}\right)^{2/\epsilon } \epsilon }{2-\epsilon}
	\approx 
	\frac{\epsilon }{2 e}+\frac{\epsilon ^2}{8 e} + \dots
\end{equation*}
is completely independent of the other two variables and depends only on $\epsilon=\xi/\nu$.
Similarly, the position of the global maximum is approximately constant as it is approximately independent of
all three variables.

In the third region where $\xi \gg 1$, the function decreases exponentially but our approximation breaks down in this regime.

\subsection{Extension to more general Kraus maps \label{app_general_kraus}}
The above formulas straightforwardly generalise to higher Kraus rank the following way.
For example, the single qubit depolarisng channel corresponds to $K_k = 3$ and in this
case there will be $\nu' = 3 \nu$ different single error events that can occur with
probabilities $\epsilon' = \epsilon/3$. The circuit error rate $\xi$ is invariant under
this transformation as $\xi = \epsilon' \nu' = \epsilon \nu$ and the upper bound
$f(\xi) \approx [\xi e^{-\xi} ]^2 /\nu \rightarrow  f(\xi)/3$ is only different by a global
constant factor. We can similarly generalise this model to other Kraus maps. 

Another simplification we have made is that we have assumed in Eq.~\eqref{error_channel}
that all gates have identical error probabilities $\epsilon$. We can extend these Kraus maps
in which all gates $\phi_k$ have possibly different error probabilities $\epsilon_k$.
In this case our previous bounds straightforwardly apply by upper bounding $\epsilon_k \leq \max \epsilon_k$
and using the largest error probability in the bound. Our results thus still hold via
the upper bound function $f(\xi') \approx [\xi' e^{-\xi'} ]^2 /\nu $, where set $\xi' := \max \epsilon_k \nu$.
Indeed, one can straightforwardly tighten these bounds by assuming some average error rate $\epsilon_{mean}$. 
It is interesting to note that the probabilities of $k$ errors happening in the circuit 
are still expected to be Poisson distributed via the Le Cam theorem even
if we allow different per-gate error probabilities $\epsilon_k$ for every gate
assuming the limit of a large number of gates (and bounded $\xi$) as discussed in Sec.~IV in \cite{cai2020multi}.

Furthermore, if a fraction $\kappa$ of the error Kraus operators $M_k$ 
from Eq.~\eqref{error_channel} (assuming Kraus rank $K=1$ for ease of notation)
commutes with the corresponding ideal unitary gates $U_k$, then we can simplify
the error model the following way. A fraction $1-\kappa$ of the gates are noise-free
while a fraction $\kappa$ of the gates undergo a higher error rate $2\epsilon$.
Thus our upper bounds still apply via the modifications as
$\nu' =  (1-\kappa) \nu$ and we can use the general upper bound on the probabilities $\epsilon \leq 2 \epsilon$ (assuming a small fraction $\kappa$).
This modifies our upper bounds as $f(\xi) \approx [\xi' e^{-\xi'} ]^2 /\nu' \rightarrow 4 (1-\kappa) f[2(1-\kappa) \xi] $,
which is generally a rescaling of the function by a multiplication of the argument $\xi$ and a multiplication by a global constant.
Indeed we observe in Fig.~\ref{plot_noisemodel_table} that the numerical data is slightly shifted to the
right when compared to the upper bounds (solid lines). This discrepancy could be explained by the fact that
a fraction of the gates in the simulated circuits actually commute with the noise Kraus operators.

\section{Details of the numerical simulations \label{app_numerics}}

\subsection{Figure~\ref{plot_eigvals} and Figure~\ref{plot_commutators}}

Random states were generated the following way. The pure state $| \psi_{id} \rangle$
was generated uniformly randomly with respect to the Haar measure of the $d$-dimensional
unitary group. Note that in all figures the dimension $d$ has been randomly generated -- these
random dimensions correspond to qubit systems when $d=2^N$ or more general qudit systems 
when $d$ is not a power of two.
The noise state $\rho_{err}$ was generated randomly by uniformly randomly
generating a $d$-dimensional probability vector $\underline{p}$ by uniformly randomly generating
points in a $d-1$-dimensional simplex. This probability vector is then used to define the eigenvalues of
$\rho_{err}$. In the next step a unitary $U$ was uniformly randomly generated with respect to
the Haar measure and $\rho_{err}$ was obtained via the transformation $U \mathrm{diag}(\underline{p}) U^\dagger$.
The noisy random state was then obtained via $\rho = \eta | \psi_{id} \rangle \langle \psi_{id} | + (1-\eta) \rho_{err}$,
where $\eta$ was uniformly randomly generated for the linear plots and uniformly randomly generated in logarithmic
scale for the logarithmic plots. 

\subsection{Figure~\ref{plot_noisemodel} and Figure~\ref{plot_noisemodel_table}}

Circuits were randomly generated by randomly selecting 200 gates from a pool of single qubit $X$ and $Z$ rotations 
(with rotation angles $\theta_k$),
and CNOT or $XX$ entangling gates. The gates are followed by either depolarising, dephasing or damping noise whose
per-gate error probabilities are fixed $\epsilon$ and were generated randomly for each circuit variant.
In case of the single qubit rotations and the $XX$ entangling gates the rotation angles $\theta_k$ we generated
according to different patterns: Rotation angles were uniformly randomly generated as $\theta_k \in (-\pi, \pi)$
in Fig.~\ref{plot_noisemodel} and in Fig.~\ref{plot_noisemodel_table} (first column), rotation angles were
increased linearly as $\theta_k = 0.01 k$ Fig.~\ref{plot_noisemodel_table} (second column)
and constant rotation angles were set as $\theta_k = 0.2$ in Fig.~\ref{plot_noisemodel_table} (third column).
	
\end{document}